\documentclass[11pt]{article}
\usepackage{amsmath,amssymb,amsthm,bm}
\usepackage{mathtools}
\usepackage{verbatim}
\usepackage{graphicx}
\usepackage{color}
\usepackage{natbib}
\usepackage{amssymb,amsmath,euscript,bbm,xcolor}
\usepackage{esint}
\usepackage{natbib}
\usepackage{float}
\usepackage[export]{adjustbox}
\usepackage{caption}
\usepackage{subcaption}
\usepackage{nicematrix}
\usepackage{algorithm}
\usepackage{algpseudocode}
\usepackage{booktabs}
\usepackage[title]{appendix}
\RequirePackage[colorlinks,citecolor=blue,urlcolor=blue]{hyperref}

\allowdisplaybreaks

\newcommand{\R}{\mathbb{R}}

\newcommand{\E}{{\mathbb{E}}}
\newcommand{\Var}{{\rm Var}}

\newcommand{\Cor}{{\rm Cor}}

\newcommand{\GOI}{{\rm GOI}}

\newcommand\numberthis{\addtocounter{equation}{1}\tag{\theequation}}

\def\R{{\mathbb R}}

\def\E{{\mathbb E}}

\algblock{Input}{EndInput}
\algnotext{EndInput}
\algblock{Output}{EndOutput}
\algnotext{EndOutput}
\newcommand{\Desc}[2]{\State \makebox[5em][l]{#1}#2}

\newtheorem{thm}{Theorem}
\newtheorem{lemma}{Lemma}
\newtheorem{prop}{Proposition}

\theoremstyle{definition}

\theoremstyle{remark}

\newtheorem{example}{Example}

\begin{document}

\title {On the peak height distribution of non-stationary Gaussian random fields: 1D general covariance and scale space}
\author{Yu Zhao$^1$ \and Dan Cheng$^3$ \and Samuel Davenport$^1$ \and Armin Schwartzman$^{1,2}$}
\date{%
    $^1$Division of Biostatistics, \\ Herbert Wertheim School of Public Health and Human Longevity Science, \\
University of California San Diego, 9500 Gilman Dr., La Jolla, CA 92093, USA \\%
    $^2$Halicioǧlu Data Science Institute, \\
University of California San Diego, 9500 Gilman Dr., La Jolla, CA 92093, USA \\%
    $^3$School of Mathematical and Statistical Sciences, \\ Arizona State University, 900 S. Palm Walk, Tempe, AZ 85281, USA \\[2ex]%
}
\maketitle

\begin{abstract}
    We study the peak height distribution of certain non-stationary Gaussian random fields. The explicit peak height distribution of smooth, non-stationary Gaussian processes in 1D with general covariance is derived. The formula is determined by two parameters, each of which has a clear statistical meaning. For multidimensional non-stationary Gaussian random fields, we generalize these results to the setting of scale space fields, which play an important role in peak detection by helping to handle peaks of different spatial extents. We demonstrate that these properties not only offer a better interpretation of the scale space field but also simplify the computation of the peak height distribution. Finally, two efficient numerical algorithms are proposed as a general solution for computing the peak height distribution of smooth multidimensional Gaussian random fields in applications.

\end{abstract}


\section{Introduction}

Peak inference is an important statistical problem in random field theory (RFT) and has been applied to multiple domains, such as neuroimaging (\citealp{Worsley2004}; \citealp{Schwartzman2019}) and statistical cosmology (\citealp{BERTSCHINGER2001}). Researchers have made significant efforts to derive formulas for the peak height distribution (\citealp{Lindgren1972}; \citealp{AZAIS2008}; \citealp{Extreme}; \citealp{Bernoulli}). The currently known exact formulas for the peak height distribution in 1D, 2D, and 3D hold for isotropic fields (\citealp{Bernoulli}) or apply to nonstationary fields that are obtained from isotropic fields by diffeomorphic transformations of the domain (\citealp{CHENG2020108672}). However, the formulas do not apply to general stationary fields, let alone general nonstationary ones. Even stationarity can be difficult to check (\citealp{Eklund2016}) and does not hold in many application settings (\citealp{Worsley1996Scale}; \citealp{Foody2004}). In order to apply RFT based peak inference more generally, there is a need to calculate the peak height distribution under non-stationarity.

Let $\{X(t): t \in \mathbbm{R}^N\}$ be a real-valued, $N$ dimensional, $C^2$ Gaussian random field. Throughout this paper, we always assume $X(t)$ satisfies the conditions (C1) and (C2) in \citet{Extreme} unless stated otherwise. For stationary random fields, the statistical properties, including the mean, variance, and covariance between points remain the same under translation across the domain. Any spatially varying statistical property leads to non-stationarity. For example, $X(t)$ is non-stationary if the mean function $\mu(t) = \E[X(t)]$ is non-constant. Non-centered Gaussian random fields (i.e. with non-zero mean function) are frequently considered in data applications (\citealp{Hayasaka2007}; \citealp{Annals}), and their peak height distribution has been studied in various RFT-based power calculation works (\citealp{Durnez16};  \citealp{Zhao2023}). In this paper, we only consider centered Gaussian random fields (i.e. with zero mean function), and focus on non-stationary Gaussian random fields with general covariance structures, including non-constant variance, spatially varying correlation, or both.

\begin{sloppypar}
The scale space random field (\citealp{Siegmund1995}) is an example of a non-stationary random field with a spatially varying correlation and has a number of real data applications (\citealp{BERTSCHINGER2001}; \citealp{Worsley1996}). One setting in which the scale space appears is as follows. In peak detection, kernel smoothing is commonly applied to improve the signal-to-noise ratio. To maximize the detection power, the kernel bandwidth should be chosen to match the size of the kernel to that of the peak. However, when various peaks have different shapes and sizes, there does not exist a single smoothing bandwidth to match all the peaks. To handle peaks of different spatial extents, it is natural to consider the bandwidth as an extra parameter and search for peaks in the scale-location space. The resulting random field which is defined in the scale-location space is called the scale space random field. Since the smoothness of the field varies with the bandwidth of the kernel, the scale space random field is non-stationary even if it is stationary at any fixed scale. In this paper, we study the scale space Gaussian field as a prototypical example of non-stationarity.
\end{sloppypar}

For Gaussian random fields, the Kac-Rice formula (see  Theorem 11.2.1 of \citealp{Adler2007}) is a powerful tool to calculate the expected number of critical points and the peak height distribution. With the Kac-Rice formula, it can be shown (see \citealp{Extreme}) that the peak height distribution of $X(t)$, i.e. the probability that the height of $X(t)$ at point $t$ exceeds the fixed threshold $u$ given that the point $t$ is a local maximum of $X(t)$, can be calculated as
\begin{equation}\label{eq:point}
\begin{split}
F_t(u)&=\frac{\E[|{\rm det} \nabla^2 X(t)|\mathbbm{1}_{\{\nabla^2 X(t) \prec 0\}}\mathbbm{1}_{\{X(t)>u\}} | \nabla X(t)=0]}{\E[|{\rm det} \nabla^2 X(t)|\mathbbm{1}_{\{\nabla^2 X(t) \prec 0\}} | \nabla X(t)=0]}.
\end{split}
\end{equation} 
In this paper, we first derive the exact formula for the peak height density (the derivative of \eqref{eq:point}) of 1D ($N = 1$) centered, smooth Gaussian processes with general covariance and compare it with that for stationary Gaussian processes (\citealp{Bernoulli}) and unit-variance Gaussian processes (\citealp{Cheng2023}). In practical applications it is often necessary to consider random fields in higher dimensions. For example, RFT based brain imaging analyses use 3D random fields due to the nature of fMRI data (see \citealp{Worsley1992}; \citealp{Worsley1996}). For dimensions $N >1$, as we shall see, the Kac-Rice formula involves computing the conditional expectation of the determinant of the Hessian matrix of $X(t)$, which is a Gaussian random matrix. There is no analytical solution to this without making further restrictive assumptions. In particular, the explicit evaluation of the peak height distribution of isotropic Gaussian random fields has been studied in \citet{Bernoulli} with the help of Gaussian random matrix theory, but there is no exact formula for general non-isotropic Gaussian random fields. 

One approach to address this issue is to connect the random field of interest $X(t)$ to some other random field $Y(t)$ with known peak height distribution. For example, \citet{CHENG2020108672} have proved that the peak height distribution remains the same under diffeomorphic transformations of the domain, i.e. $X(t) = Y(f(t))$ where $f$ is a diffeomorphic map. Another approach is to simplify the calculation by making practical assumptions, for example, assuming the field has constant  spatial smoothness (\citealp{Eklund2016}). However this assumption is still rather restrictive and is likely to be violated in real data. 

In this paper, we propose to compute the peak height distribution using numerical approximations. We introduce two convenient algorithms that compute the peak height distribution by evaluating the Kac-Rice formula numerically, which work for any sufficiently smooth Gaussian random field. When considering the choice between the two options, it is worth noting that the first algorithm is faster but tends to be less accurate for large thresholds. The second algorithm, on the other hand, 
is slower but achieves better accuracy for large thresholds. Overall, these algorithms achieve higher accuracy with much less computing time compared to the direct Monte Carlo simulation approach, which simulates the field itself. 

The rest of the paper is organized as follows. In Section \ref{sec:1d_nonstationary}, we derive the general formula for the peak height distribution of non-stationary Gaussian processes in 1D with general covariance and show examples. In Section \ref{sec:scale_space_field}, we consider the peak height distribution of multidimensional non-stationary Gaussian random fields, discuss the properties of the scale space field, and prove the peak height distribution of the scale space field remains the same across the domain. Section \ref{sec:num_alg} presents two numerical algorithms that can be used to compute the peak height distribution of smooth Gaussian random fields. In Section \ref{sec:simulation}, we validate the formulas and properties derived in this paper and evaluate the performance of the numerical algorithms in simulation.

Software to implement the peak height distribution is available in the RFTtoolbox package (\citealp{RFTtoolbox}) and code to reproduce the analyses and simulation results is available at \url{https://github.com/YuZhao-UCSD/NonstationaryGRF}.

\section{1D non-stationary Gaussian processes with general covariance}\label{sec:1d_nonstationary}

In this section, we start from the simplest non-stationary Gaussian random field, the 1D non-stationary Gaussian process with general covariance. Our aim is to derive the explicit form of the peak height distribution by evaluating the Kac-Rice formula \eqref{eq:point}, and compare it with the peak height distribution of stationary Gaussian processes. 

Assuming that the random process \{$X(t): t\in \mathbbm{R}$\} is non-stationary, the covariance between any two points may depend on both their location and the distance between them. As a generalization of the covariance function under the stationarity assumption (\citealp{AZAIS2008}), we write the covariance function of a non-stationary Gaussian process as
\[
\E[X(t - d/2)X(t + d/2)] = h(t,\tau = d^2)
\]
for an appropriate function $h(\cdot):[0,\infty) \to \mathbbm{R}$. Denote
\begin{equation} \label{def:cov_derv}
    h_0(t) = h(t,0), \quad h'_1(t) = \left.\frac{\partial h}{\partial t}  \right|_{\tau = 0}, \quad \left.h'_2(t) = \frac{\partial h}{\partial \tau}\right|_{\tau = 0}, \quad 
    \left.h''_{12}(t) = \frac{\partial^2 h}{\partial t \partial \tau}\right|_{\tau = 0},\: \text{etc}.
\end{equation}
The variance-covariance matrix of $(X(t),\: X'(t),\: X''(t))$ can be derived using (5.5.5) in \citet{Adler2007}. The lemma below is a generalization of Lemma 3.2 in \citet{Bernoulli} for non-stationary Gaussian processes. 

\begin{lemma}\label{lem:X_cov}
Let \{$X(t), t \in \mathbb{R}$\} be a centered, smooth 1D Gaussian process. Then for each $t \in \mathbbm{R}$,
\begin{align*}
    \Var(X(t)) & = h_0(t), \\
    \Var(X'(t)) & = \frac{1}{4}h''_{11}(t) - 2h'_2(t), \\
    \Var(X''(t)) & = \frac{1}{16}h''''_{1111}(t) - h'''_{112}(t) + 12h''_{22}(t), \\
    \E[X(t)X'(t)] &= \frac{1}{2}h'_1(t),\\
    \E[X(t)X''(t)] &= \frac{1}{4}h''_{11}(t) + 2h'_2(t), \\
    \E[X'(t)X''(t)] &= \frac{1}{8}h'''_{111}(t) - h''_{12}(t).
\end{align*}
\end{lemma}
Note that for stationary processes, $h(t, \tau) = h(\tau)$ is constant with respect to $t$. In this case, all the derivatives above with respect to the first argument vanish and the result reduces exactly to Lemma 3.2 in \citet{Bernoulli}.

\subsection{Peak height distribution of non-stationary Gaussian processes with general covariance}


To evaluate the Kac-Rice formula \eqref{eq:point}, we need to derive the variances and covariances of the first and second order derivatives of $X(t)$, which are provided by Lemma \ref{lem:X_cov}. Detailed derivation can be found in the proof. The following theorem states a compactly elegant formula for the peak height distribution of centered, smooth Gaussian processes. Let $\phi(t)$ and $\Phi(t)$ denote the standard normal density and CDF, respectively.
\begin{thm}\label{thm:peak_height_1D}
 The peak height density of a smooth Gaussian process \{$X(t), t \in \mathbb{R}$\} with mean 0 and standard deviation $\sigma(t)$ is given by 
 \begin{align}\label{eq:peak_height_1D}
 f_t(x) & = \frac{1}{\tilde{\sigma}(t)}\phi\left(\frac{x}{\tilde{\sigma}(t)}\right) \sqrt{2\pi(1-\rho^2(t))}\psi\left(\frac{-\rho(t)x}{\sqrt{1-\rho^2(t)}\tilde{\sigma}(t)}\right),
\numberthis   
 \end{align}
 where $\tilde{\sigma}^2(t): = \Var(X(t)|X'(t) = 0)$, $\rho(t): = \Cor(X(t),X''(t)|X'(t) = 0)$ satisfying $|\rho(t)|<1$, and the function $\psi(\cdot)$ is defined as
\begin{equation}
\psi(x) = \int_{-\infty}^x \Phi(y)dy = \phi(x) + x\Phi(x), \quad x\in \R.
\label{def: psi}
\end{equation}
The mean and variance of the peak height are $-\sqrt{\pi/2}\rho(t)\tilde{\sigma}(t)$ and $[ 1-(\pi/2-1)\rho^2(t) ]\tilde{\sigma}^2(t)$ respectively.
\end{thm}

Surprisingly, the density function \eqref{eq:peak_height_1D} only depends on two parameters, $\rho(t)$ and $\tilde{\sigma}(t)$, and has the form of a Gaussian density multiplied by a tilting factor. By Lemma \ref{lem:X_cov}, $\rho(t)$ can be computed as
\begin{align*}
    \rho(t) = &  \Cor(X(t),X''(t)|X'(t) = 0) \\
     = & \frac{\E[X(t)X''(t)]\Var(X'(t)) - \E[X(t)X'(t)]\E[X'(t)X''(t)]}{\sqrt{\Var(X(t))\Var(X'(t)) - \E[X(t)X'(t)]^2}} \\
     & \times \frac{1}{\sqrt{\Var(X'(t))\Var(X''(t)) - \E[X'(t)X''(t)]^2} }\numberthis \label{eq:rho_vcov}\\
     = & -\frac{2\left(2h'_{2}(t)-\frac{1}{4}h''_{11}(t)\right)\left(2h'_{2}(t)+\frac{1}{4}h''_{11}(t)\right)-h'_{1}(t)\left(h''_{12}(t)-\frac{1}{8}h'''_{111}\right)}{\sqrt{-{\left(h''_{12}(t)-\frac{1}{8}h'''_{111}(t)\right)}^2-\left(2h'_{2}(t)-\frac{1}{4}h''_{11}(t)\right)\left(12h''_{22}(t)-h'''_{112}(t)+\frac{1}{16}h''''_{1111}(t)\right)}} \\ 
     & \times \frac{1}{\sqrt{-{h'_{1}(t)}^2-4h_0(t)\left(2h'_{2}(t)-\frac{1}{4}h''_{11}(t)\right)}}. \numberthis \label{eq:rho} 
\end{align*}
Both \eqref{eq:rho_vcov} and \eqref{eq:rho} can be used to compute $\rho(t)$, depending on which is more convenient, the moments or the partial derivatives of the covariance function. In the analytical examples below, it is easier to use \eqref{eq:rho_vcov}. In the numerical methods (Section \ref{sec:num_alg}), it is easier to use \eqref{eq:rho} because all the terms in \eqref{eq:rho} can be programmed as the partial derivatives of the covariance function. The same reasoning can be applied to the computation of $\tilde{\sigma}(t)$: 
\begin{align}\label{eq:sigma_t_vcov}
    \tilde{\sigma}^2(t) = & \Var(X(t)|X'(t) = 0)
     = \Var(X(t)) - \frac{\E[X(t)X'(t)]^2}{\Var(X'(t))}  
     = h_0(t) - \frac{h'^2_1(t)}{h''_{11}(t) - 8h'_2(t)}.
\end{align}

Compared to the parameter $\kappa$ of the peak height density derived in \cite{Bernoulli}, using $\rho(t)$ and $\tilde{\sigma}(t)$ to characterize the peak height density improves the understanding and interpretation of the density function, as they both have a clear statistical meaning. Figure \ref{fig:1D_PeakHeightDensity} displays how these two parameters affect the peak height density. As we can see from \eqref{eq:peak_height_1D} and Figure \ref{fig:PeakHeightDensityRho}, if we fix $\tilde{\sigma}(t)$, the density functions with $\rho(t)$ and $-\rho(t)$ are reflection symmetric with respect to $x = 0$. As $\rho(t)$ increases, the absolute value of the mean increases proportionally, and the variance decreases. The parameter $\tilde{\sigma}(t)$ is the scale parameter of the peak height distribution. The distribution is more spread out for large $\tilde{\sigma}(t)$ as displayed in Figure \ref{fig:PeakHeightDensitySigma}. Increasing $\tilde{\sigma}(t)$ leads to a linear increase in both the absolute value of the mean and standard deviation. Special cases of these parameters, including when $X(t)$ is stationary, are described in Section \ref{sec:special_case} below.

\begin{figure}
\centering
\begin{subfigure}[a]{200pt}
    \centering
    \includegraphics[width = 200pt]{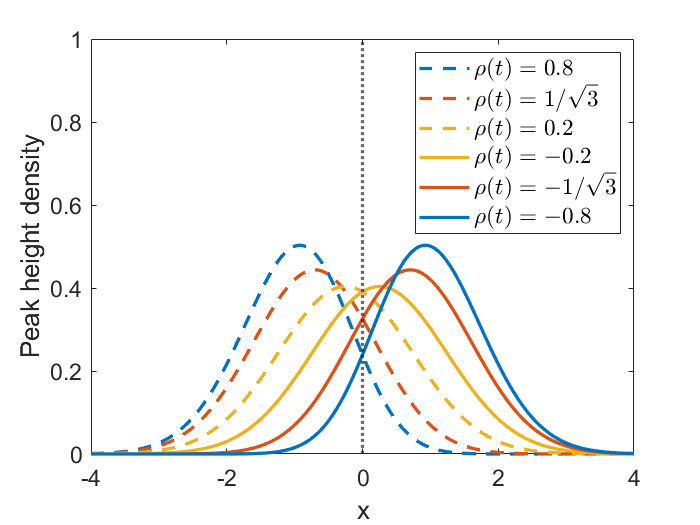}
    \caption{$\tilde{\sigma}(t) = 1$}
    \label{fig:PeakHeightDensityRho}
\end{subfigure}
\begin{subfigure}[a]{200pt}
    \centering
    \includegraphics[width = 200pt]{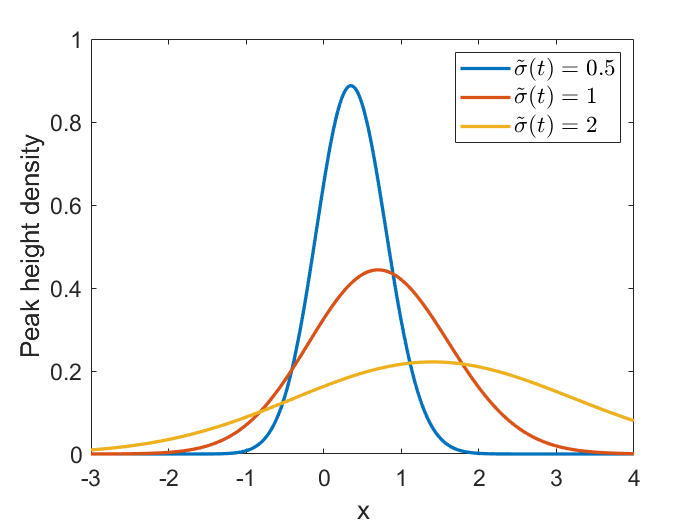}
    \caption{$\rho(t) = -1/\sqrt{3}$}
    \label{fig:PeakHeightDensitySigma}
\end{subfigure}
\caption{The peak height density \eqref{eq:peak_height_1D} with different parameters. Left panel: Fixing $\tilde{\sigma}(t)$ at 1, the effect of $\rho(t)$ on the peak height density. Right panel: Fixing $\rho(t)$ at $-1/\sqrt{3}$, the effect of $\tilde{\sigma}(t)$ on the peak height density. }
\label{fig:1D_PeakHeightDensity}
\end{figure}

Theorem \ref{thm:peak_height_1D} can be used to compute the peak height distribution of any centered, smooth, non-stationary Gaussian process, including the Gaussian process with non-constant variance. In general, Gaussian processes with non-constant variance are useful in modeling the heteroscedastic noise (i.e. the variance of noise depends on location), which is often present in real-world problems (see e.g. \citealp{le2005}; \citealp{lazaro2011}). Understanding the peak height distribution of such Gaussian processes is crucial for peak detection given heteroscedastic noise.

\begin{example}\label{ex:non_const_var}
Consider the following Gaussian process with non-constant variance
\begin{equation*}
    X(t) = \sigma(t)\int_{-\infty}^\infty \frac{1}{\sqrt{\nu}} k\left(\frac{t-s}{\nu}\right)\,dB(s),
\end{equation*}
where $k(s)$, $s \in \mathbbm{R}$, is a smooth kernel with $\int k^2(s) ds = 1$, $dB(s)$ is Gaussian white noise and $\nu$ is the bandwidth. Figure \ref{fig:1D_variance_example_X} displays simulated instances of this process on a grid ranging from 0 to 1 with a step size of 0.005. The process is generated by multiplying a stationary, unit-variance Gaussian process (convolution of Gaussian white noise with a Gaussian kernel of bandwidth $\nu = 0.3$) by a linear standard deviation function $\sigma(t) = t + 0.1$. Consequently, the spatial correlation is independent of the location $t$, as shown in Figure \ref{fig:1D_variance_example_cor}. Figure \ref{fig:1D_variance_example_sigma_tilde} and Figure \ref{fig:1D_variance_example_rho} display the two parameters $\rho(t)$ and $\tilde{\sigma}(t)$ that affect the peak height distribution. Peak height distribution of this process at a fixed location $t$ can be computed by plugging $\rho(t)$ and $\tilde{\sigma}(t)$ into \eqref{eq:peak_height_1D}.
\end{example}

\begin{figure}[ht]
\centering
\begin{subfigure}[a]{200pt}
    \centering
    \includegraphics[width = 200pt]{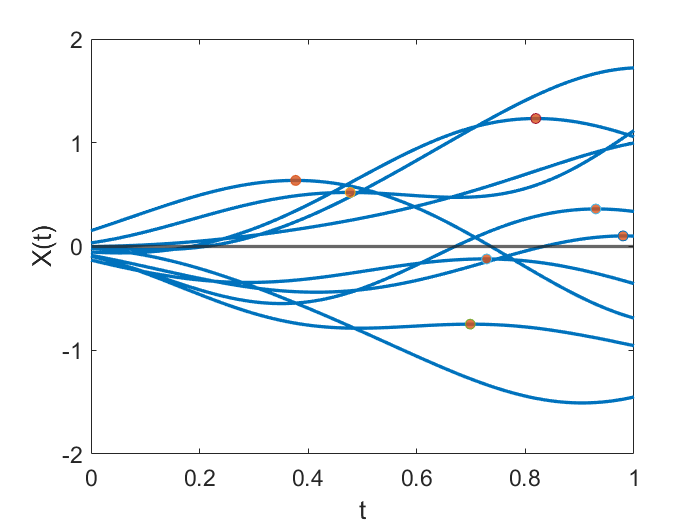}
    \caption{Simulated instances with orange dots indicating the local maxima}
    \label{fig:1D_variance_example_X}
\end{subfigure}
\begin{subfigure}[a]{200pt}
    \centering
    \includegraphics[width = 200pt]{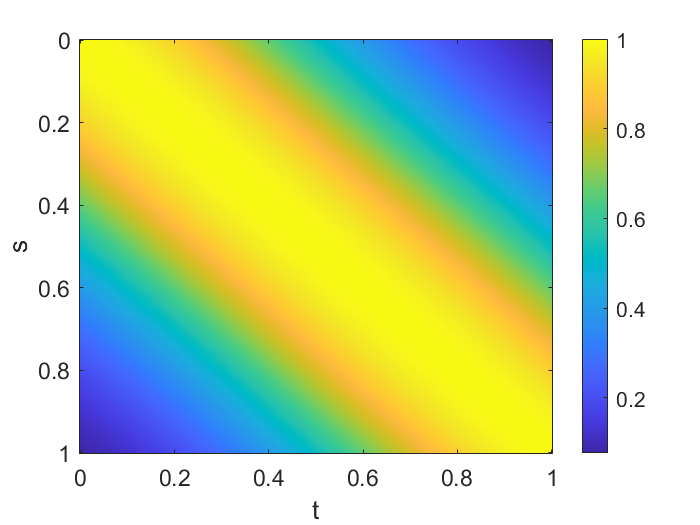}
    \caption{Correlation}
    \label{fig:1D_variance_example_cor}
\end{subfigure}
\begin{subfigure}[a]{200pt}
    \centering
    \includegraphics[width = 200pt]{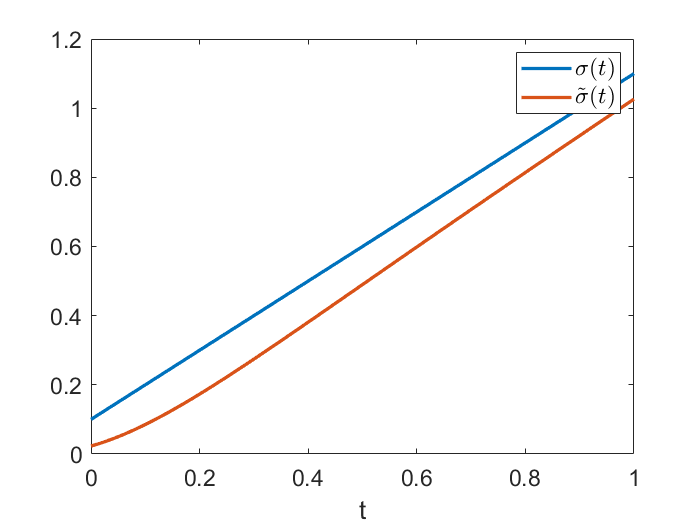}
    \caption{$\sigma(t)$ and $\tilde{\sigma}(t)$}
    \label{fig:1D_variance_example_sigma_tilde}
\end{subfigure}
\begin{subfigure}[a]{200pt}
    \centering
    \includegraphics[width = 200pt]{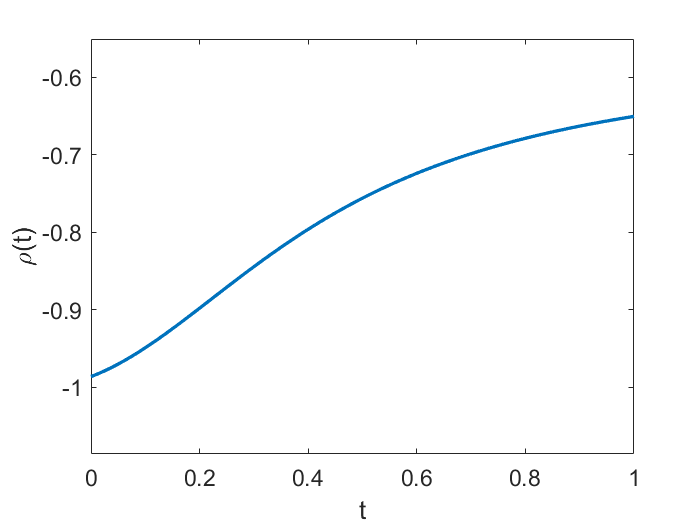}
    \caption{$\rho(t)$}
    \label{fig:1D_variance_example_rho}
\end{subfigure}
\caption{Simulated instances of the Gaussian process with non-constant standard deviation function $\sigma(t) = t + 0.1$, its spatial correlation, and the parameters $\sigma(t)$, $\tilde{\sigma}(t)$, and $\rho(t)$.}
\label{fig:1D_variance_example}
\end{figure}

\subsection{Special cases}\label{sec:special_case}
In particular, if $X(t)$ has constant variance, $\rho(t)$ and $\tilde{\sigma}(t)$ can be simplified as
\begin{align*}
    \rho(t) &= -\frac{\Var{(X'(t))}}{\sqrt{\Var{(X''(t))} - \frac{\E[X'(t)X''(t)]^2}{\Var{(X'(t))}}}} = \frac{2h'_2(t)}{\sqrt{-h'''_{112}(t) + 12h''_{22}(t)+\frac{h''^2_{12}(t)}{2h'_{2}(t)}}}, \numberthis \label{eq:rho_constant_sigma_vcov}\\
    \tilde{\sigma}^2(t) & = \Var(X(t))= h_0.
\end{align*}
This special case has been considered in \citet{Cheng2023}, where the Gaussian process is assumed to have unit variance.  Note that \eqref{eq:rho_constant_sigma_vcov} indicates $\rho(t)\leq 0$ when $X(t)$ has constant variance. 

If we further assume stationarity and unit-variance, then $\rho(t)$ and $\tilde{\sigma}(t)$ become constant:
\begin{align*}
    \rho &= -\frac{\Var{(X'(t))}}{\sqrt{\Var{(X''(t))} - \frac{\E[X'(t)X''(t)]^2}{\Var{(X'(t))}}}}= \frac{h'_2}{\sqrt{3h''_{22}}}, \\
    \tilde{\sigma}^2 & = \Var(X(t))= 1. 
\end{align*}
As a special case of \eqref{eq:rho_constant_sigma_vcov}, the parameter $\rho$ is non-positive for stationary random processes. 

The general peak height density \eqref{eq:peak_height_1D} has a similar form as the peak height density of a unit-variance, stationary Gaussian process derived in \citet{Bernoulli}:
\begin{align*}
        f_t(x) 
        &= \frac{\sqrt{3 -\kappa^2}}{\sqrt{3}}\phi\left(\frac{\sqrt{3 }x}{\sqrt{3 -\kappa^2}}\right) + \frac{\sqrt{2\pi}\kappa x}{\sqrt{3}} \Phi\left(\frac{\kappa x}{\sqrt{3 -\kappa^2}}\right)\phi(x) \\
        & = \phi(x) \sqrt{\frac{2\pi(3-\kappa^2)}{3}}\psi\left(\frac{\kappa x}{\sqrt{3-\kappa^2}}\right),
\numberthis
\label{eqn:density_1d}
\end{align*} 
where $\kappa$ is defined as
\begin{equation*}
\kappa:= -h'_2/\sqrt{h''_{22}} = -\sqrt{3}\rho.
\end{equation*}
Note that $\kappa = 1$ if the covariance function is Gaussian (\citealp{Bernoulli}), and the corresponding $\rho = -1/\sqrt{3}$.

\subsection{The boundary case: $|\rho(t)| = 1$}
Theorem \ref{thm:peak_height_1D} above required that $|\rho(t)| < 1$. The random vector $(X(t),X'(t),X''(t))$ is degenerate when $|\rho(t)| = 1$ violating the assumption (C2) in \citet{Extreme}. Therefore, the technique used in the proof of Theorem \ref{thm:peak_height_1D} needs to be modified for the boundary case. 


\begin{prop}\label{prop:boundary}


When $\rho(t) = \pm 1$, the peak height distribution of $X(t)$ is $\mp {\rm Rayleigh}(\tilde{\sigma}(t))$ with density
\begin{equation}\label{eq:peak_density_degen_n}
    f_t(x) = \mp \frac{\sqrt{2\pi}x}{\tilde{\sigma}^2(t)}\phi\left(\frac{x}{\tilde{\sigma}(t)}\right).
\end{equation}
The mean and variance are $\mp \sqrt{\pi/2}\tilde{\sigma}(t)$ and $( 2-\pi/2)\tilde{\sigma}^2(t)$ respectively.
\end{prop}


\begin{example}\label{ex:cos}
Define the Cosine process
\begin{equation}\label{def:cos}
   X(t) = c_1z_1\cos(\omega t) + c_2z_2\sin(\omega t),
\end{equation}
where $z_1$ and $z_2$ are independent, standard Gaussian random variables and $c_1$, $c_2$, and $\omega$ are positive constants. $X(t)$ has non-constant variance $\sigma^2(t) = c_1^2\sin^2(\omega t) + c_2^2 \cos^2(\omega t)$ and is therefore non-stationary. It can be derived by \eqref{eq:rho_vcov} and \eqref{eq:sigma_t_vcov} that $\rho(t) = -1$ and $\tilde{\sigma}(t) = c_1c_2/\sigma(t)$. The peak height density is Rayleigh
\begin{equation*}
    f_t(x) = \frac{\sqrt{2\pi}[c_1^2\sin^2(\omega t) + c_2^2 \cos^2(\omega t)]x}{c_1^2c_2^2}\phi\left(\frac{\sqrt{c_1^2\sin^2(\omega t) + c_2^2 \cos^2(\omega t)}x}{c_1c_2}\right), \quad x \geq 0.
\end{equation*}
Figure \ref{fig:1D_cos_example_X} shows simulated instances of the cosine process ($\omega = 2$, $c_1 = 3$ and $c_2 = 4$) on a grid ranging from 0 to $2\pi$ with a step size of $0.002\pi$. In Figure \ref{fig:1D_cos_example_X_peak}, we compare the peak height of this process near $t = \pi/4$ (large $\tilde{\sigma}(t)$) and $t = \pi/2$ (small $\tilde{\sigma}(t)$). The parameter $\tilde{\sigma}(t)$ is the only parameter of the peak height distribution, and has a positive effect on both the mean and standard deviation, as displayed in Figure \ref{fig:1D_cos_example_peak}.
\end{example}

\begin{figure}[!ht]
\centering
\begin{subfigure}[a]{200pt}
    \centering
    \includegraphics[width = 200pt]{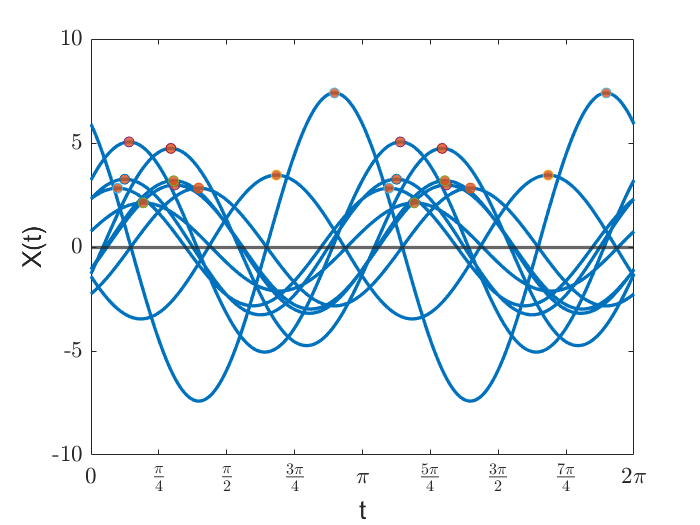}
    \caption{Simulated instances with orange dots indicating the local maxima}
    \label{fig:1D_cos_example_X}
\end{subfigure}
\begin{subfigure}[a]{200pt}
    \centering
    \includegraphics[width = 200pt]{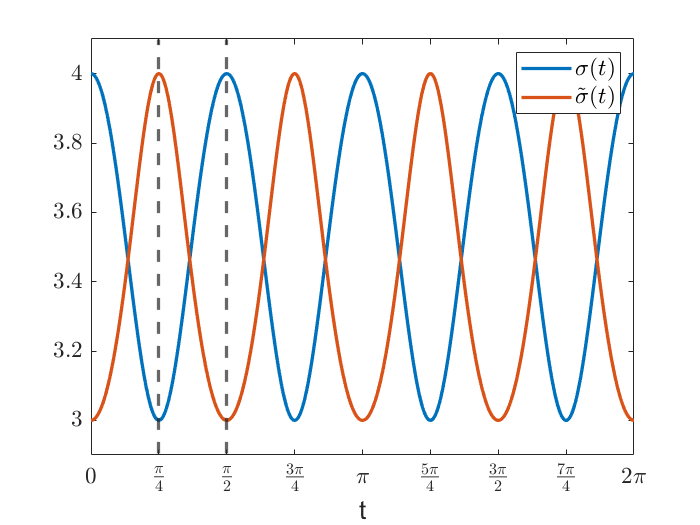}
    \caption{$\sigma(t)$ and $\tilde{\sigma}(t)$}
    \label{fig:1D_cos_sigma_t}
\end{subfigure}
\begin{subfigure}[a]{200pt}
    \centering
    \includegraphics[width = 200pt]{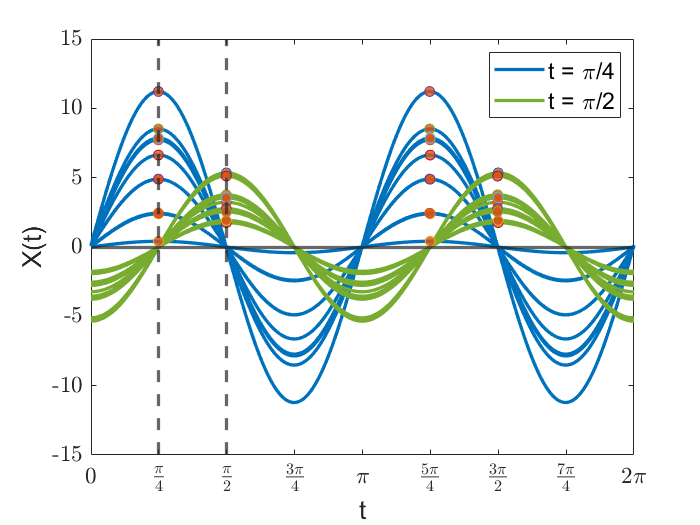}
    \caption{Selected simulated instances having a peak near $t = \pi/4$ versus those having a peak near $t = \pi/2$ } \label{fig:1D_cos_example_X_peak}
\end{subfigure}
\begin{subfigure}[a]{200pt}
    \centering
    \includegraphics[width = 200pt]{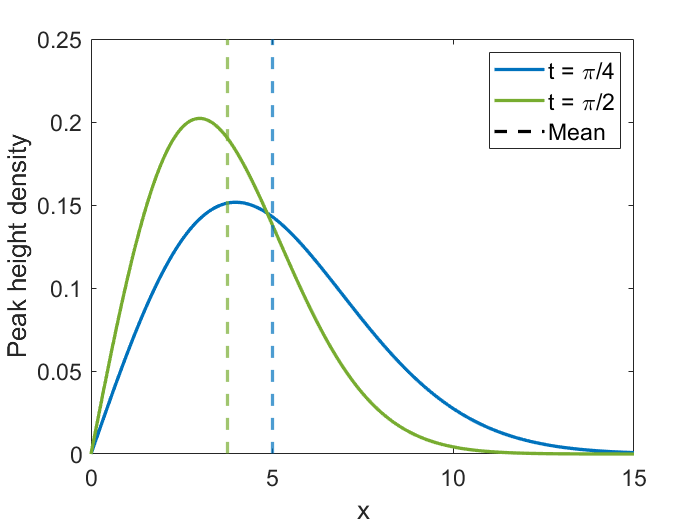}
    \caption{Peak height distribution at $t = \pi/4$ versus $t = \pi/2$}
    \label{fig:1D_cos_example_peak}
\end{subfigure}
\caption{Simulated instances of the Cosine process in Example \ref{ex:cos}, $\sigma(t)$, $\tilde{\sigma}(t)$ and the peak height distribution at $t = \pi/4$ (large $\tilde{\sigma}(t)$) versus $t = \pi/2$ (small $\tilde{\sigma}(t)$).}
\label{fig:1D_cos_example}
\end{figure}





\subsection{The Gaussian process with non-constant bandwidth}\label{sec:non_const_band}

The parameter $\rho(t)$ is the only parameter that controls how the peak height distribution varies over the domain when the process $X(t)$ has constant variance. While constant $\rho(t)$ is necessary for stationarity, non-stationarity does not necessarily imply non-constant $\rho(t)$. Next, we provide an interesting example of a non-stationary Gaussian process that has constant $\rho(t)$. 

We have discussed the Gaussian process with non-constant variance. Another type of non-stationary Gaussian process is the Gaussian process with spatially varying corelation. Gaussian processes with spatially varying correlation are commonly used (see e.g. \citealp{paciorek2003}; \citealp{remes2017}) to model data whose smoothness varies with $t$. 

By the spectral representation theorem (\citealp{Adler2007}), every mean square continuous, centered, stationary Gaussian process can be represented as a convolution of a white noise process with a proper kernel function that specifies the correlation structure:
\begin{equation}\label{def:Z_t}
    Z(t) = \int_{-\infty}^\infty \frac{1}{\sqrt{\nu}} k\left(\frac{t-s}{\nu}\right)\,dB(s) ,
\end{equation}
where $k(s)$, $s \in \mathbbm{R}$, is a smooth kernel with $\int k^2(s) ds = 1$ so that $Z(t)$ is unit-variance, $dB(s)$ is Gaussian white noise and $\nu$ is the bandwidth. In \eqref{def:Z_t}, the kernel bandwidth $\nu$ controls the smoothness and spatial correlation. We define the Gaussian process with non-constant bandwidth as a generalization of \eqref{def:Z_t} 
\begin{align*}
    X(t) &= \int_{-\infty}^\infty \frac{1}{\sqrt{\nu(t)}} k\left(\frac{t-s}{\nu(t)}\right)\,dB(s) ,
\numberthis
\label{def:non_constant_band}
\end{align*}
where the bandwidth function $\nu(t)$, $t\in \R$, is a positive non-constant smooth function. The Gaussian process with non-constant bandwidth is non-stationary since the spatial correlation varies with $t$. Also, since this process has unit variance, the peak height distribution depends solely on $\rho(t)$. In the following theorem, we show that even though the process is non-stationary, $\rho(t)$ is constant when the bandwidth $\nu(t)$ is linear in $t$. 

\begin{thm}\label{thm:linear_nu}
\begin{sloppypar}
Consider the non-stationary Gaussian process
\begin{equation*}
    X(t) = \int_{-\infty}^\infty \frac{1}{\sqrt{\nu(t)}} k\left(\frac{t-s}{\nu(t)}\right)\,dB(s), 
\end{equation*}
where $k(s)$, $s \in \mathbbm{R}$, is a smooth kernel with $\int k^2(s) ds = 1$. If the bandwidth function $\nu(t)$ is linear in $t$, then the parameter $\rho(t)= \Cor(X(t),X''(t)|X'(t) = 0)$ is independent of $t$.
\end{sloppypar}
\end{thm}

Theorem \ref{thm:linear_nu} also indicates the peak height distribution is independent of $t$ when the bandwidth function $\nu(t)$ is linear. This property is particularly interesting, considering the process $X(t)$ is non-stationary, and helps to reduce the computation time of evaluating the peak height distribution over a certain domain. 

\subsection{The Gaussian process with non-constant bandwidth: Gaussian kernel}\label{sec:non_constant_bandwidth_Gauss}

As an example, we demonstrate the explicit calculation of $\rho(t)$ when the kernel function $k(\cdot)$ in \eqref{def:non_constant_band} is Gaussian, i.e.  
\begin{align}
    X(t) &= \int_{-\infty}^\infty \frac{\sqrt{2}\pi^{1/4}}{\sqrt{\nu(t)}} \phi\left(\frac{t-s}{\nu(t)}\right)\,dB(s),
\label{def:non_constant_band_Gaussian}
\end{align}
where $\phi(x)$ is the standard Gaussian density. 

Since $X(t)$ is unit-variance, $\tilde{\sigma}(t) = 1$. Applying \eqref{eq:rho_constant_sigma_vcov}, we have
\begin{align}
\rho(t) & = -\frac{1}{2}\sqrt{\frac{(\nu'^2(t)+1)^{3}}{\splitfrac{\frac{43}{16}\nu'^6(t) + (\frac{103}{16}-\frac{1}{2}\nu''(t)\nu(t))\nu'^4(t) + \frac{9}{2}\nu'^2(t)+\frac{3}{4}}{+\nu''(t)\nu(t)+\frac{1}{2}\nu''^2(t)\nu^2(t)}}}. 
\label{eq:rho_gauss}
\end{align}

It is straightforward from \eqref{eq:rho_gauss} that when $\nu(t)$ is linear in $t$, say $\nu(t) = \nu't + \nu_0$, $\rho(t)$ is constant:
\begin{equation*}
\rho(t) = -\frac{1}{2}\sqrt{\frac{(\nu'^2+1)^{3}}{\frac{43}{16}\nu'^6 + \frac{103}{16}\nu'^4 + \frac{9}{2}\nu'^2+\frac{3}{4}}}.
\end{equation*}

\begin{example}\label{ex:non_const_band}
Consider a special case of \eqref{def:non_constant_band} when the kernel function $k(\cdot)$ is Gaussian and 
\begin{equation*}
     \nu(t) = 0.5t + 0.1.
    \end{equation*}
   By \eqref{eq:rho_gauss}, the resulting $\rho(t) = -2/\sqrt{19}$. In Figure \ref{fig:1D_bandwidth_example_X}, we show simulated instances of the process on a grid ranging from 0 to 1 with a step size of 0.005. From Figure \ref{fig:1D_bandwdith_example_cor}, we can clearly see the spatial correlation is dependent on the location. More specifically, since $\nu(t)$ is monotone increasing, the process becomes smoother as $t$ increases, as illustrated in the simulated instances and the correlation plot. Furthermore, it is important to note that in this case, $\rho(t) = -2/\sqrt{19}$ is independent of $t$ (see Figure \ref{fig:1D_bandwdith_example_rho}) as stipulated by Theorem \ref{thm:linear_nu}.
\end{example}

\begin{figure}[ht]
\centering
\begin{subfigure}[a]{200pt}
    \centering
    \includegraphics[width = 200pt]{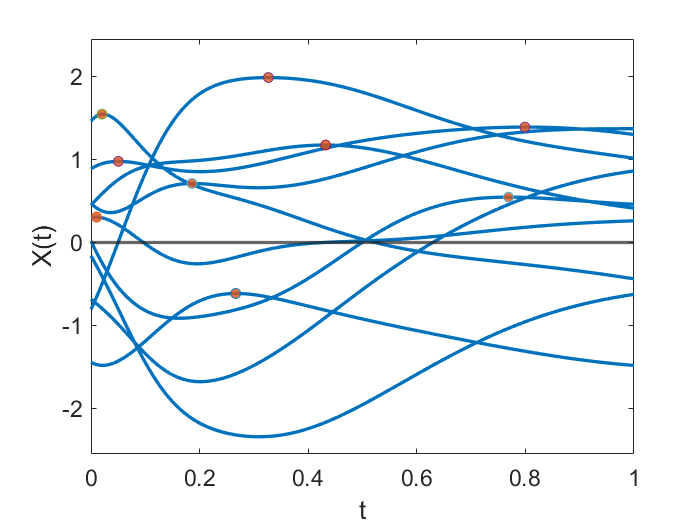}
    \caption{Simulated instances with orange dots indicating the local maxima}
    \label{fig:1D_bandwidth_example_X}
\end{subfigure}
\begin{subfigure}[a]{200pt}
    \centering
    \includegraphics[width = 200pt]{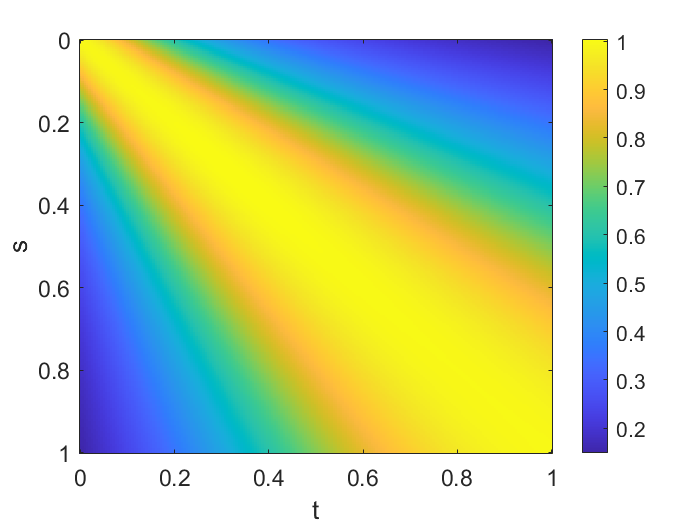}
    \caption{Correlation}
    \label{fig:1D_bandwdith_example_cor}
\end{subfigure}
\begin{subfigure}[a]{200pt}
    \centering
    \includegraphics[width = 200pt]{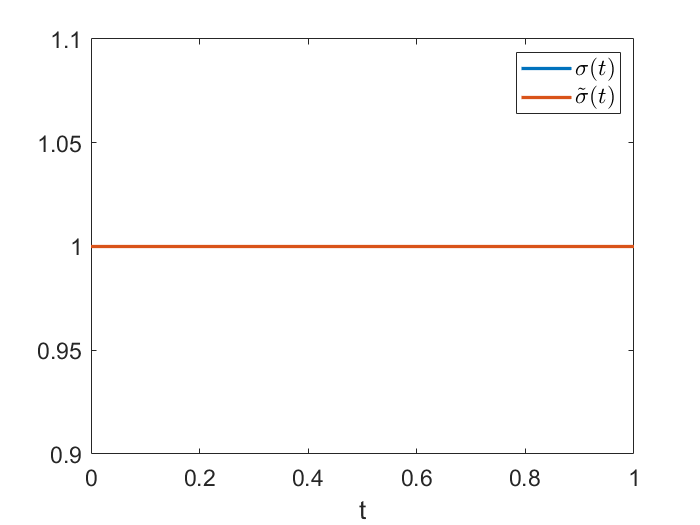}
    \caption{$\sigma(t)$ and $\tilde{\sigma}(t)$}
    \label{fig:1D_bandwdith_example_sigma_t}
\end{subfigure}
\begin{subfigure}[a]{200pt}
    \centering
    \includegraphics[width = 200pt]{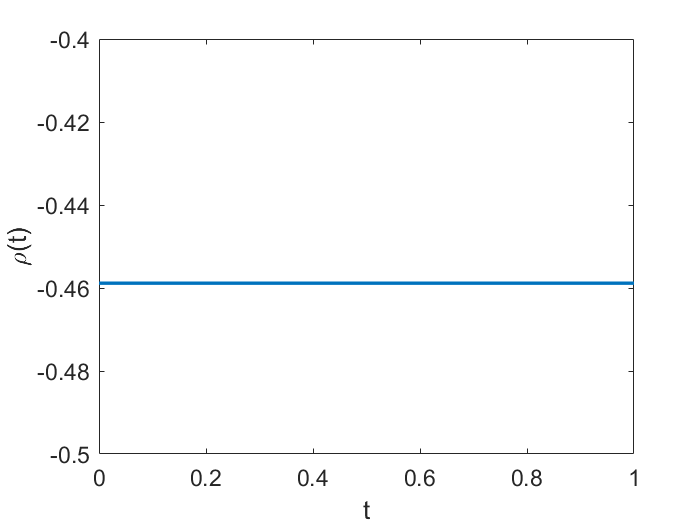}
    \caption{$\rho(t)$}
    \label{fig:1D_bandwdith_example_rho}
\end{subfigure}
\caption{Simulated instances of the Gaussian process with non-constant bandwidth (Gaussian kernel) described in Example \ref{ex:non_const_band}, its spatial correlation, and the parameters $\sigma(t)$, $\tilde{\sigma}(t)$, and $\rho(t)$.}
\label{fig:1D_bandwidth_example}
\end{figure} 

\subsection{The Gaussian process with non-constant bandwidth and variance}\label{sec:non_const_band_var}

By \eqref{eq:rho_constant_sigma_vcov}, the parameter $\rho(t)$ is non-positive when the process has constant variance, but can be positive in general. To take a step further, we study the Gaussian process with non-constant bandwidth and variance and give an example of positive $\rho(t)$. We generalize \eqref{def:non_constant_band} to have non-constant variance:
\begin{align}
    Y(t) &= \sigma(t)X(t)\numberthis
\label{def:non_constant_band_var_XY} = \sigma(t) \int_{-\infty}^\infty \frac{1}{\sqrt{\nu(t)}} k\left(\frac{t-s}{\nu(t)}\right)\,dB(s) ,
\end{align}
where the standard deviation function $\sigma(t)$, $t\in \R$, is a positive non-constant smooth function so that $\Var(Y(t)) = \sigma^2(t)$. Taking first and second derivatives on both sides of \eqref{def:non_constant_band_var_XY}, we have $Y'(t)=\sigma'(t)X(t) + \sigma(t)X'(t)$ and $Y''(t)=\sigma''(t)X(t) + 2\sigma'(t)X'(t) + \sigma(t)X''(t)$. We use the variances and covariances of the first and second order derivatives of $X(t)$ (see \ref{sec:non_constant_bandwidth_Gauss}) as a bridge for computing those of $Y(t)$:
\begin{align*}
    \Var(Y'(t)) =&  \sigma'^2(t) + \sigma^2(t)\Var(X'(t)), \\
    \Var(Y''(t)) =&  \sigma''^2(t) + (4\sigma'^2(t)-2\sigma(t)\sigma''(t))\Var(X'(t)) \\
    & + \sigma^2(t)\Var(X''(t))+ 4\sigma(t)\sigma'(t)\E[X'(t)X''(t)], \\
    \E[Y(t)Y'(t)] = & \sigma(t)\sigma'(t), \\
    \E[Y(t)Y''(t)] = & \sigma(t)\sigma''(t) - \sigma^2(t)\Var(X'(t)), \\
    \E[Y'(t)Y''(t)] = & \sigma'(t)\sigma''(t) + \sigma(t)\sigma'(t)\Var(X'(t)) + \sigma^2(t)\E[X'(t)X''(t)].
\end{align*}
Then the parameter $\rho(t)$, $\tilde{\sigma}(t)$, and the peak height distribution can be obtained by applying \eqref{eq:rho_vcov}, \eqref{eq:sigma_t_vcov}, and \eqref{eq:peak_height_1D}.

\begin{example}\label{ex:non_const_var_band}
We consider the following specification of $Y(t)$ in \eqref{def:non_constant_band_var_XY}
\begin{align*}
    k(s) & = \sqrt{2}\pi^{1/4}\phi(s), \\ 
    \nu(t) &= 0.5t + 0.1, \\
    \sigma(t) & = 8t^2 - 10t + 6, \numberthis \label{eq:non_constant_band_nu_sigma_t}
\end{align*}
 to give an example of positive $\rho(t)$. In Figure \ref{fig:1D_bandwidth_sigma_example_X}, we show simulated instances of this process on a grid ranging from 0 to 1 with a step size of 0.005. As displayed in Figure \ref{fig:1D_bandwdith_sigma_example_rho}, $\rho(t)$ reaches its minimum -0.54 at around $t = 0.2$ and maximum 0.38 at around $t = 0.7$. Taking a closer look at $t = 0.2$ and $0.7$, we compare the peak height distribution at $t = 0.2$ (negative $\rho(t)$) versus $t = 0.7$ (positive $\rho(t)$) in Figure \ref{fig:1D_bandwidth_sigma_example_density}. Given that the process has zero mean for all $t$, and in particular at $t = 0.7$, it is quite surprising that both the mean and median of the peak height distribution are negative. In conclusion, for a general Gaussian process, the parameter $\rho(t)$ can have either a positive or negative value depending on the specification of the process and the location $t$.
\end{example}

\begin{figure}[!ht]
\centering
\begin{subfigure}[a]{200pt}
    \centering
    \includegraphics[width = 200pt]{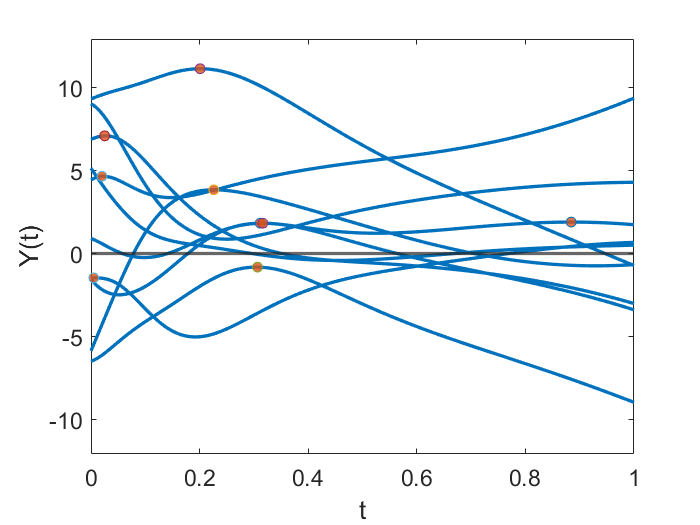}
    \caption{Simulated instances with orange dots indicating the local maxima}
    \label{fig:1D_bandwidth_sigma_example_X}
\end{subfigure}
\begin{subfigure}[a]{200pt}
    \centering
    \includegraphics[width = 200pt]{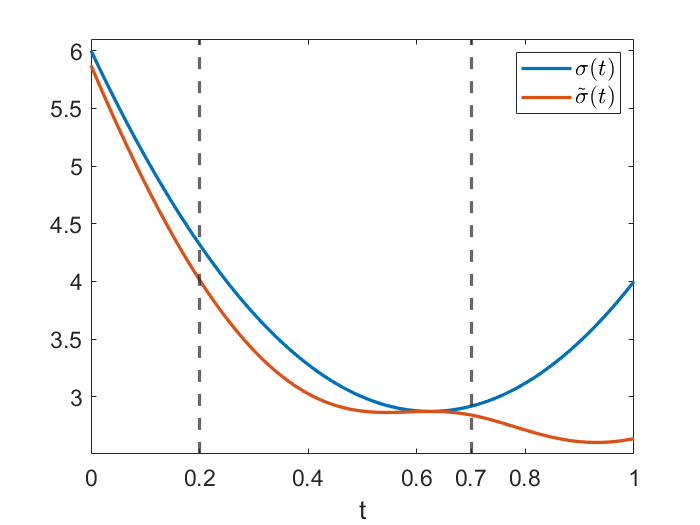}
    \caption{$\sigma(t)$ and $\tilde{\sigma}(t)$}
    \label{fig:1D_bandwdith_sigma_example_sigma_t}
\end{subfigure}
\begin{subfigure}[a]{200pt}
    \centering
    \includegraphics[width = 200pt]{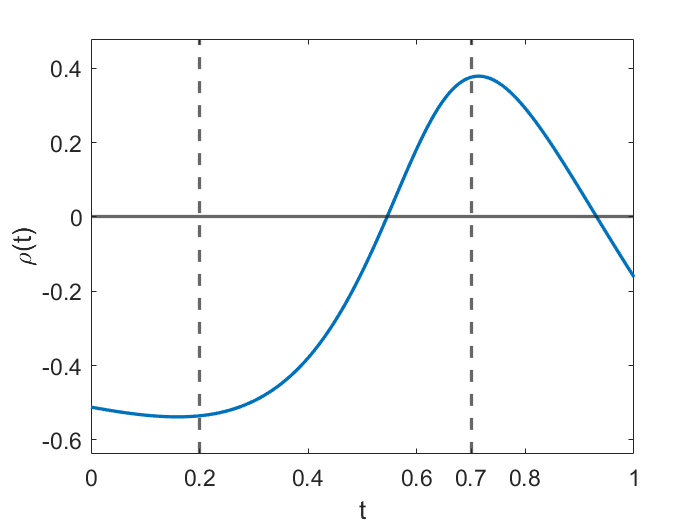}
    \caption{$\rho(t)$}
    \label{fig:1D_bandwdith_sigma_example_rho}
\end{subfigure}
\caption{Simulated instances of the Gaussian process with non-constant bandwidth and variance described in Example \ref{ex:non_const_var_band}, and the parameters $\sigma(t)$, $\tilde{\sigma}(t)$, and $\rho(t)$.}
\label{fig:1D_bandwidth_sigma_example}
\end{figure}

\begin{figure}[!ht]
\centering
\begin{subfigure}[a]{200pt}
    \centering
    \includegraphics[width = 200pt]{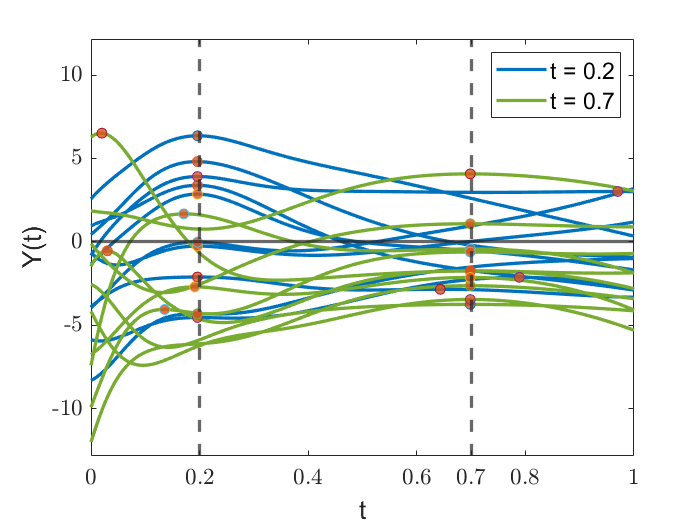}
    \caption{Selected simulated instances having a peak near $t = 0.2$ versus those having a peak near $t = 0.7$}
    \label{fig:1D_bandwidth_sigma_example_X_peak}
\end{subfigure}
\begin{subfigure}[a]{200pt}
    \centering
    \includegraphics[width = 200pt]{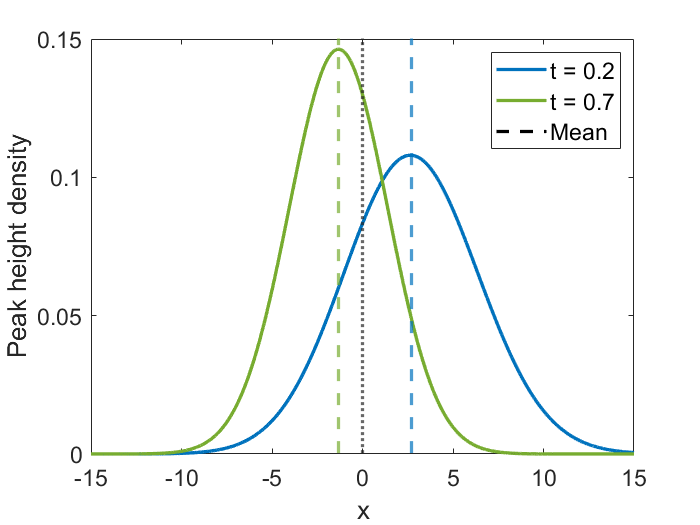}
    \caption{Peak height distribution at $t = 0.2$ versus $t = 0.7$}
    \label{fig:1D_bandwidth_sigma_example_density}
\end{subfigure}
\caption{The peak height distribution of the Gaussian process with non-constant bandwidth and variance described in Example \ref{ex:non_const_var_band} at $t = 0.2$ (negative $\rho(t)$) versus $t = 0.7$ (positive $\rho(t)$).}
\label{fig:1D_bandwidth_sigma_peak_height}
\end{figure}

\section{Multidimensional non-stationary Gaussian random fields: Scale space}\label{sec:scale_space_field}

We have explored the peak height distribution of 1D non-stationary Gaussian processes. However, many real-world problems exist in higher dimensions. As we step into the realm of higher dimensional non-stationary Gaussian random fields, the complexity of the Kac-Rice formula increases drastically, making direct evaluation difficult. 

In this section, with the results we have developed in 1D as a foundation, we explore the properties of the peak height distribution of an important multidimensional non-stationary Gaussian random field: the scale space Gaussian random field. Although these properties are not enough to allow us to derive the explicit formula for the peak height distribution, they are helpful for a better understanding of the spatial structure of these non-stationary random fields. Moreover, these properties effectively increase the efficiency of the numerical methods we highlight later in this paper (see Section \ref{sec:simulation}).

\subsection{The scale space Gaussian random field}

First, we consider the following zero-mean unit-variance 2D scale space Gaussian random field
\begin{equation}
X(t,\nu) = \int \frac{1}{\sqrt{\nu}}k\left(\frac{s-t}{\nu}\right) d B(s),
\label{eqn:scale_field_2d}
\end{equation}
where $k(s)$, $s \in R$, is a smooth kernel with $\int k^2(s) ds = 1$. In the scale space field, the smoothing bandwidth is not fixed and is treated as a parameter, so that the parameter space $(t,\nu)$ is 2D. Solving the peak height distribution of a scale space field helps to detect peaks of unknown location and scale. 

If we draw a slice plane for the 2D scale space field, the slice along the plane is a 1D Gaussian process with linearly increasing bandwidth, the same process as the special case studied in Section \ref{sec:non_const_band} above. An example of the 2D scale space field and the slice plane is displayed in Figure \ref{fig:2d_scale_slice}, and the slice along the plane is displayed in Figure \ref{fig:2d_scale_slice_1d}. The 2D scale space field can be seen as an extension of the Gaussian process with non-constant bandwidth. 

For the Gaussian process with non-constant bandwidth, we have demonstrated that when the bandwidth $\nu(t)$ is a linear function of $t$, the peak height distribution remains the same over the domain. Therefore, the peak height distribution of the 1D process obtained by slicing a 2D scale space field with a flat plane is independent of $t$. Next, we generalize this result to 2D scale space fields and, even further, to scale space fields of any dimension.

\begin{figure}
\centering
\begin{subfigure}[a]{200pt}
    \centering
    \includegraphics[width = 200pt]{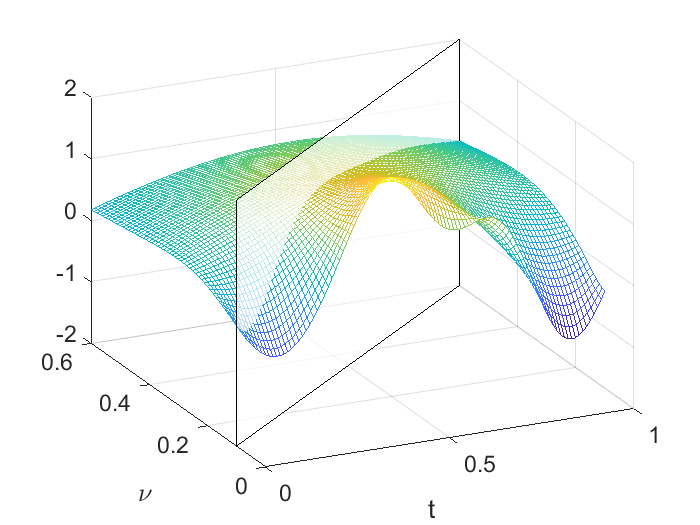}
    \caption{$X(t,\nu)$ and a slice plane $\nu(t) = 0.5t + 0.1$}
    \label{fig:2d_scale_slice}
\end{subfigure}
\begin{subfigure}[a]{200pt}
    \centering
    \includegraphics[width = 200pt]{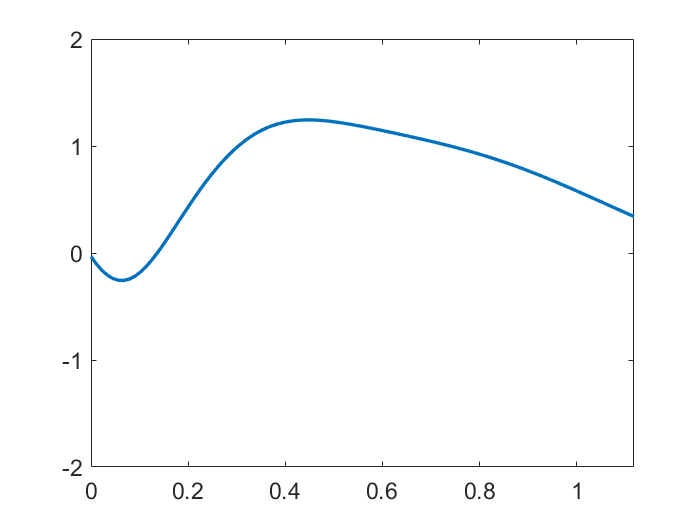}
    \caption{Slice along the plane}
    \label{fig:2d_scale_slice_1d}
\end{subfigure}
\caption{A single instance of the 2D scale space field and a slice plane $\nu(t) = 0.5t + 0.1$.}
\label{fig:2d_scale_example}
\end{figure}

To evaluate the Kac-Rice formula and compute the peak height distribution of the scale space field, we once again need the variances and covariances of the first and second order derivatives. To simplify the computation, we let $v = -\log\nu$, so that $X(t,v)$ is stationary in $v$ for any fixed $t$, and $X(t,v)$ is also stationary in $t$ for any fixed $v$ (\citealp{Siegmund1995}). We replace $\nu$ in \eqref{eqn:scale_field_2d} with $e^{-v}$ and the covariance of $X$ after the change of variable is 
\begin{equation}\label{eq:cov_scale}
\E{[X(t_1,v_1)X(t_2,v_2)]} = e^{(v_1+v_2)/2} \int k((s-t_1)e^{v_1}) k((s-t_2)e^{v_2}) ds.
\end{equation}
Define
\begin{align*}
\nabla X &= \left(\frac{\partial X}{\partial t}, \frac{\partial X}{\partial v}\right)^T, \quad \nabla^2 X =  \begin{pmatrix}
\frac{\partial^2 X}{\partial t^2} & \frac{\partial^2 X}{\partial t \partial v}\\
\frac{\partial^2 X}{\partial t \partial v} & \frac{\partial^2 X}{\partial v^2}
\end{pmatrix},
\end{align*}
\noindent and their joint distribution
\begin{equation*}
\begin{pmatrix}
\nabla X \\
{\rm vec}(\nabla^2 X)
\end{pmatrix} \sim
N\left(0,\begin{pmatrix}
\Sigma_{11} & \Sigma_{12}\\
\Sigma_{21} & \Sigma_{22}
\end{pmatrix}\right).
\end{equation*}

The variances and covariances can be obtained using (5.5.5) in \citet{Adler2007}.
\begin{equation*}
\Sigma_{11}  = \begin{pmatrix}
e^{2v}\int k'^2(s) ds & 0 \\
0 & \int (\frac{1}{2}k(s) + sk'^2(s)) ds
\end{pmatrix},
\end{equation*}
\begin{equation*}
\Sigma_{22} = \begin{pmatrix}
e^{4v}\int k''^2(s) ds &  e^{2v}\int k''(s) \eta(s) ds & 0\\
e^{2v}\int k''(s) \eta(s) ds & \int \eta^2(s) ds & 0 \\
0 & 0 & e^{2v} \int (\frac{3}{2}k'(s)+sk''^2(s)) ds
\end{pmatrix},
\end{equation*}
\begin{equation*}
\Sigma_{21}  = 
\begin{pmatrix}
0 & e^{2v}\int k''(s)(\frac{1}{2}k(s)+sk'(s)) ds \\
0 & \int (\frac{1}{2}k(s)+sk'(s))\eta(s) ds \\
e^{2v} \int (\frac{3}{2}k'(s)+sk''(s))(\frac{1}{2}k(s)+sk'(s)) ds& 0
\end{pmatrix},
\end{equation*}
where $\eta(s) = 1/4k(s) + 2sk'(s) + s^2k''(s)$. Using the properties of conditional normal distribution, we obtain
\begin{equation*}
\nabla^2 X| \nabla X = 0 \sim N(0, \Sigma_{22}-\Sigma_{21}\Sigma_{11}^{-1}\Sigma_{12}),
\label{eqn:con_dist_den}
\end{equation*}
and
\begin{equation*}
\nabla^2 X| X = x, \nabla X = 0 \sim N(\tilde{\Sigma}_{21} \tilde{\Sigma}_{11}^{-1}(x,0,0)^T, \Sigma_{22}-\tilde{\Sigma}_{21}\tilde{\Sigma}_{11}^{-1}\tilde{\Sigma}_{12}),
\end{equation*}
where
\begin{equation*}
\tilde{\Sigma}_{11}  = \begin{pmatrix}
1 & 0 \\
0 & \Sigma_{11}
\end{pmatrix}, \qquad
\tilde{\Sigma}_{21} =        
    \begin{pNiceArray}{cc}[margin]
    -e^{2v}\int k'^2(s) ds & \Block{3-1}{\Sigma_{21}} \\
    \frac{1}{4} - \int s^2 k'^2(s) ds  & \\
    0 & \\
    \end{pNiceArray}.
\end{equation*}
\subsection{The scale space Gaussian random field: Gaussian kernel} \label{sec:scale_gaussian}

Consider the special case when the kernel function is Gaussian, i.e.
\begin{equation*}
k(s) = \sqrt{2}\pi^{1/4}\phi(s). 
\end{equation*}
The covariance matrices can be further simplified as 
\begin{equation*}\label{eqn:scale_gaussian_cov}
\Sigma_{11}  = \begin{pmatrix}
\frac{1}{2}e^{2v} & 0 \\
0 & \frac{1}{2}
\end{pmatrix}, \quad
\Sigma_{22} = \begin{pmatrix}
\frac{3}{4}e^{4v} &  \frac{1}{4}e^{2v} & 0\\
\frac{1}{4}e^{2v} & \frac{7}{4} & 0 \\
0 & 0 & \frac{5}{4}e^{2v}
\end{pmatrix}, \quad
\Sigma_{21}  = 
\begin{pmatrix}
0 & -\frac{1}{2}e^{2v}\\
0 & 0 \\
\frac{1}{2}e^{2v} & 0
\end{pmatrix},
\end{equation*}
and
\begin{equation*}
\nabla^2 X | \nabla X = 0 \sim N \left(0,
\begin{pmatrix}
\frac{1}{4}e^{4v} & \frac{1}{4}e^{2v} & 0\\
\frac{1}{4}e^{2v} & \frac{7}{4} & 0\\
0 & 0 & \frac{3}{4}e^{2v}
\end{pmatrix}\right).
\label{eqn:scale_gaussian_X''}
\end{equation*}

A pattern can be observed from these matrices: the variance of the derivative with respect to $v$ is constant, and the variance of the derivative with respect to $t$ has the form of a constant multiplied by $e^{2v}$. The second-order derivatives show a similar pattern. This pattern is crucial for proving the peak height distribution of the scale space field independent of $t$ and $v$.



\subsection{Peak height distribution of the scale space Gaussian random field}
The following result is a generalization of Theorem \ref{thm:linear_nu} for the scale space field.
\begin{thm}
\label{thm:scale_space}
Consider the (N+1)-dimensional scale space Gaussian random field
\begin{equation}
X(t,\nu) = \int \frac{1}{\sqrt{\nu^N}}k\left(\frac{s-t}{\nu}\right) d B(s),
\label{eqn:scale_field}
\end{equation}
where $k(s)$, $s \in R^{N}$, is an isotropic smooth kernel with $\int k^2(s) ds = 1$. The peak height distribution of $X(t,\nu)$ is independent of $(t,\nu)$. 
\end{thm}

Theorem \ref{thm:scale_space} demonstrates the peak height distribution of the scale space field remains the same across the domain although it is non-stationary. This discovery is helpful both in theoretical calculation of the peak height distribution and in application. In theoretical calculation of the peak height distribution, we can choose some specific $\nu$ that helps to simplify the evaluation of the Kac-Rice formula, as the final result is independent of $\nu$. In application, this property vastly reduces the running time, particularly when dealing with multidimensional data like fMRI images.

\section{Numerical implementation} \label{sec:num_alg}

Gaussian random field theory is widely used as a statistical framework in the analysis of fMRI images (\citealp{Friston1994}; \citealp{Worsley1996}). As we mentioned, the main bottleneck for peak inference in fMRI analysis using RFT is the model assumptions particularly stationarity. Making the stationarity assumption is common in practice (\citealp{Chumbley2010}; \citealp{Annals}) to simplify the derivation of the analytical formula for the peak height distribution, but has been shown not to hold in practice (\citealp{Eklund2016}). For datasets for which is it not reasonable to assume stationarity we must explore a more general solution.
\subsection{Algorithm}

In Section \ref{sec:1d_nonstationary}, we derived the peak height distribution formula for 1D non-stationary Gaussian processes. When it comes to multidimensional non-isotropic and non-stationary Gaussian random fields, deriving a closed form for the peak height distribution \eqref{eq:point} is challenging. In order to evaluate these quantities we observe that they can in fact be computed numerically. In particular for each $t$, ($X(t)$, $\nabla X(t)$,  {\rm vech}$(\nabla^2 X(t))$) are jointly Gaussian. As such the expectations in the numerator and the denominator of \eqref{eq:point} are conditional expectations of Gaussian random variables. Given the covariance matrix corresponding to the joint distribution, it is thus possible to simulate from this distribution and provide a numerical estimate of the expectation. The details for this are provided in Algorithm \ref{alg:Kac-Rice-1}.  
      
\begin{algorithm}[ht]
\caption{Numerical Kac-Rice}\label{alg:Kac-Rice-1}
\begin{algorithmic}
\Input (the variances and covariances at a given location $t$)
 \Desc{fv}{the variance of the field $X$}
 \Desc{dv}{the $d \times d$ variance-covariance matrix of the first derivatives $\nabla X$}
 \Desc{d2v}{the $[d(d+1)/2] \times [d(d+1)/2]$ variance-covariance matrix of the second derivatives {\rm vech}($\nabla^2 X$)}
 \Desc{fdcov}{the $1 \times d$ covariance matrix between the field $X$ and its first derivatives $\nabla X$}
 \Desc{fd2cov}{the $1 \times [d(d+1)/2]$ covariance matrix between the field $X$ and its second derivatives {\rm vech}($\nabla^2 X$)}
 \Desc{dd2cov}{the $d \times [d(d+1)/2]$ covariance matrix between the first derivatives $\nabla X$ and second derivatives {\rm vech}($\nabla^2 X$)}
 \Desc{$u$}{a vector giving the thresholds}
 \Desc{\texttt{niters}}{the number of iterations used to calculate the expectation}
\EndInput
\Ensure the variance-covariance matrix is positive definite
\State Add $-\infty$ to the end of $u$
\State Compute the covariance matrix $\boldsymbol{\Sigma}$ of $(X,\nabla^2 X)|\nabla X = 0$ with the input variances and covariances
\State Generate i.i.d. samples ($X_i$, $H_i$) (\texttt{i = 1...niters}) from the multivariate Gaussian distribution $N(\boldsymbol{0},\boldsymbol{\Sigma})$
    \For{\texttt{i = 1...niters}}
    \State Compute $\det(H_i)$
\EndFor
\For{\texttt{j = 1...length($u$)}}
    \State Compute KR[j] $= \frac{1}{\texttt{niters}} \sum|\det (H_i)|\mathbbm{1}_{H_i \prec 0}\mathbbm{1}_{X_i > u[\text{j}]}$
\EndFor  
\State \Return KR[:-1] /  KR[-1]
\end{algorithmic}
\end{algorithm}

At high thresholds $u$, Algorithm \ref{alg:Kac-Rice-1} can be applied but will require a large number of simulations in order to calculate the peak height distribution accurately. This is due to the indicator term $\mathbbm{1}_{\{X(t) > u\}}$ in \eqref{eq:point}. As $u \to -\infty$, the indicator has no effect, so that all drawn samples contribute to the numerical approximation.  However, as $u\to \infty$, the probability that $X(t)>u$ shrinks to zero, leading to a significantly smaller effective sample size compared to the number of simulated instances. If we take the 3D scale space field as an example and set $u = 3.91$, which is the $99^\text{th}$ percentile of the peak height distribution, then the effective sample size is only about $1-\Phi(3.91) = 0.005\%$ of the total number of simulated instances. In order to achieve the same accuracy, we need to sample about $1/0.005\% = 20000$ times more data for $u = 3.91$ compared to $u = -\infty$ which is computationally intensive. In Algorithm \ref{alg:Kac-Rice-2}, we solve this issue by computing the numerator of \eqref{eq:point} as
\begin{align*}
   & \E[|{\rm det} \nabla^2 X(t)|\mathbbm{1}_{\{\nabla^2 X(t) \prec 0\}}\mathbbm{1}_{\{X(t)>u\}} | \nabla X(t)=0] \\
   = &  \E[|{\rm det} \nabla^2 X(t)|\mathbbm{1}_{\{\nabla^2 X(t) \prec 0\}} | X(t)>u, \nabla X(t)=0] p(X(t) > u | \nabla X(t)=0).
\end{align*}
We generate $X(t)$ from a truncated normal distribution to bound $X(t)$ from below by $u$, leading to a more efficient approach for large $u$.    
 
\begin{algorithm}
\caption{Numerical Kac-Rice for large u}\label{alg:Kac-Rice-2}
\begin{algorithmic}
\Input
 \State Same as Algorithm 1
\EndInput
\State Add $-\infty$ to the end of u
\State Compute the covariance matrix $\boldsymbol{\Sigma}$ of $\nabla^2 X|X,\nabla X = (0,0)$ with the input variances and covariances
\State Compute the probability density $p_{\nabla X}$ of $\nabla X$ with dv
\State Compute the tail distribution $P_{X | \nabla X = 0}$ of $X | \nabla X = 0$ with fv, dv, and fdcov
\State Generate i.i.d. sample $H_i$ (\texttt{i = 1...niters}) from the multivariate Gaussian distribution $N(\boldsymbol{0},\boldsymbol{\Sigma})$ 
\For{\texttt{j = 1...length($u$)}}
    \State Generate i.i.d. sample $\tilde{X}_i$ (\texttt{i = 1...niters}) from a truncated Gaussian distribution $X|X>u[\text{j}]$.
    \For{\texttt{i = 1...niters}}
        \State Compute $\tilde{H}_i = H_i + \E[\nabla^2 X|X, \nabla X = (\tilde{X}_i,0)]$ so that $\tilde{H}_i$ follows the same distribution as $\nabla^2 X|X, \nabla X = (\tilde{X}_i,0)$ 
        \State Compute $\det(\tilde{H}_i)$
    \EndFor
    \State Compute  KR[j] $=\frac{1}{\texttt{niters}} \sum |\det (\tilde{H}_i)|\mathbbm{1}_{\tilde{H}_i \prec 0}$
    \State Update KR[j] by multiplying $p_{\nabla X}(0)P_{X | \nabla X = 0}(\text{u}[\text{j}])$
\EndFor   
\State \Return KR[:-1] /  KR[-1]
\end{algorithmic}
\end{algorithm}

The most difficult part of evaluating the peak height distribution \eqref{eq:point} is computing the expectation of the determinant of the Hessian matrix $\nabla^2 X(t)$ given the field $X(t)$ and the gradient $\nabla X(t)$. However, deriving the joint distribution of ($X(t)$, $\nabla X(t)$, $\nabla^2 X(t)$) is relatively easy. As such our numerical algorithms get around this difficulty and make it possible to calculate the peak height distribution \eqref{eq:point}, without imposing restrictive assumptions.

\subsection{Computational efficiency} \label{sec:sim_efficiency}
If the underlying distribution of the random fields is known then an alternative relatively simple but computationally expensive means of estimating the peak height distribution is via direct simulation. This can be done by generating a sufficiently large number of instances of the random field of interest on a discrete lattice, and recording the heights of all observed local maxima from these simulated instances. The empirical distribution of the simulated peak heights can be used to approximate the true peak height distribution. Direct simulation is easy to implement but requires a significant amount of time to run. The low efficiency can be attributed to two primary factors. First, generating the field itself provides more information than we actually need to estimate the peak height distribution. Second, the performance heavily depends on the granularity and the size of the grid when we use a discrete lattice to approximate the continuous domain, and the amount of data points grows exponentially as the dimension increases. Also note that direct simulation will struggle at high $u$ levels since only a small proportion, or potentially none, of the simulated peaks exceed the threshold.

To compare the performance between direct simulation and the numerical Kac-Rice Algorithm \ref{alg:Kac-Rice-1}, we apply these two approaches to estimate the peak height distribution of a scale space field while controlling the number of simulations. In the direct simulation, the 3D scale space field is generated from \eqref{eqn:scale_field} over a grid of size 20 $\times$ 20 $\times$ 20 pixels with scale parameter $\nu$ ranging from 0.2 to 1.2. For the numerical Kac-Rice method, we apply Algorithm \ref{alg:Kac-Rice-1}. The comparison of runtime for the two methods is presented in Table \ref{tab:num_KR}, and the simulated CDF is displayed in Figure \ref{fig:sim_vs_num}. 

\begin{table}[h]
\caption{Runtime of direct simulation versus numerical Kac-Rice Algorithm \ref{alg:Kac-Rice-1} for estimating the peak height distribution of a scale space field. The total number of observed peaks in simulated instances is provided in the parenthesis.}
\centering
\begin{tabular}{l 
                c 
                c} \toprule
 \# of simulations  & 10,000  & 100,000 \\ \toprule
    Direct simulation & 1,306.42s (221 peaks)  & 13,293.85s (2441 peaks)\\
    Numerical Kac-Rice & 0.20s & 1.50s\\
\bottomrule
\end{tabular}\label{tab:num_KR}
\end{table}

\begin{figure}
\centering
\begin{subfigure}[a]{200pt}
    \centering
    \includegraphics[width = 200pt]{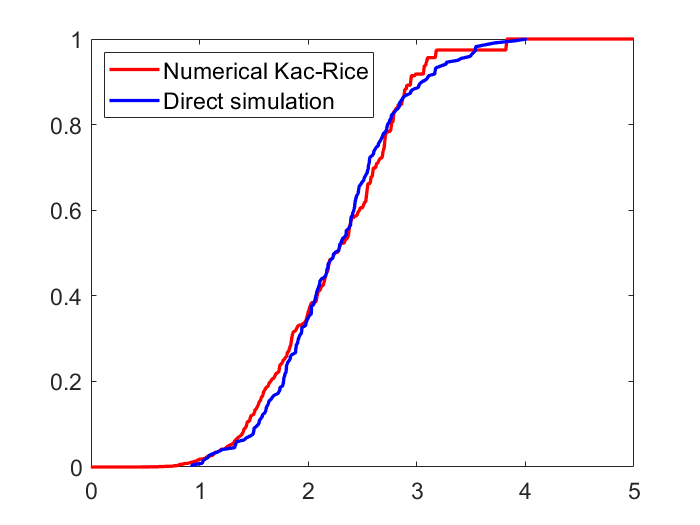}
    \caption{10,000 simulations}
    \label{fig:sim_vs_num_1e4}
\end{subfigure}
\begin{subfigure}[a]{200pt}
    \centering
    \includegraphics[width = 200pt]{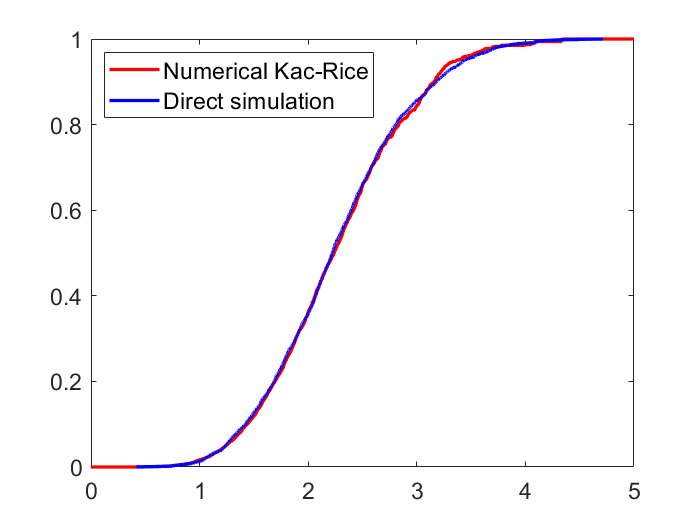}
    \caption{100,000 simulations}
    \label{fig:sim_vs_num_1e5}
\end{subfigure}
\caption{Simulated CDF of the 3D scale space field mentioned in \ref{sec:sim_efficiency}, direct simulation versus numerical Kac-Rice.}
\label{fig:sim_vs_num}
\end{figure}

As demonstrated in Table \ref{tab:num_KR} and Figure \ref{fig:sim_vs_num}, our numerical Kac-Rice algorithm performs comparably to direct simulation with a significantly lower runtime. Neither of the two methods can produce a smooth simulated CDF curve for 10,000 simulations, and when the number of simulations increases to 100,000, they both perform well. Therefore, if estimating the peak height distribution is the only goal, the numerical Kac-Rice algorithm is preferable. This is the case when we explore the theoretical properties of the peak height distribution, such as how different parameters affect the CDF.

\section{Validation via simulations} \label{sec:simulation}

In this section, we present simulation results to validate the theoretical properties of the Gaussian process with non-constant bandwidth and the scale space field. In \ref{sim:non_constant_band}, we validate Theorems \ref{thm:peak_height_1D} and \ref{thm:linear_nu} which only involves 1D Gaussian processes. Drawing samples from a 1D Gaussian process is simple and computationally efficient, so we apply direct simulation to estimate the peak height distribution and relevant parameters of interest. For the scale space field, we have demonstrated in \ref{sec:sim_efficiency} 
that generating the field directly is highly time-consuming. Therefore, in \ref{sim:3d_scale}, we apply the numerical Kac-Rice Algorithm \ref{alg:Kac-Rice-1} to validate Theorem \ref{thm:scale_space}. 

\subsection{The Gaussian process with non-constant bandwidth}\label{sim:non_constant_band}

To validate Theorems \ref{thm:peak_height_1D} and \ref{thm:linear_nu}, we generate $100,000$ Gaussian processes with non-constant bandwidth. $t$ is sampled from a grid of $[0,1]$ containing 200 points. As defined in \eqref{def:non_constant_band}, each process is obtained as the convolution of Gaussian white noise with a Gaussian kernel of non-constant bandwidth function $\nu(t) = 0.5t + 0.1$. Simulated instances of the process are displayed in Figure \ref{fig:1D_bandwidth_example_X}.

\begin{figure}
\centering
\begin{subfigure}[a]{200pt}
    \centering
    \includegraphics[width = 200pt]{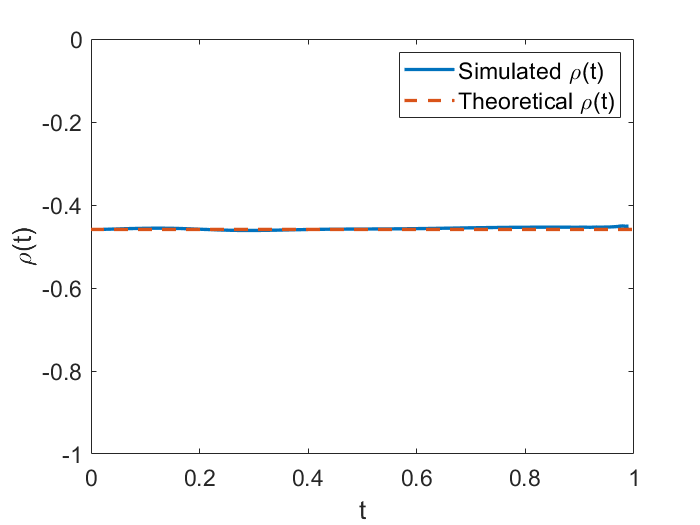}
    \caption{$\rho(t)$}
    \label{fig:1D_linear_rho}
\end{subfigure}
\begin{subfigure}[a]{200pt}
    \centering
    \includegraphics[width = 200pt]{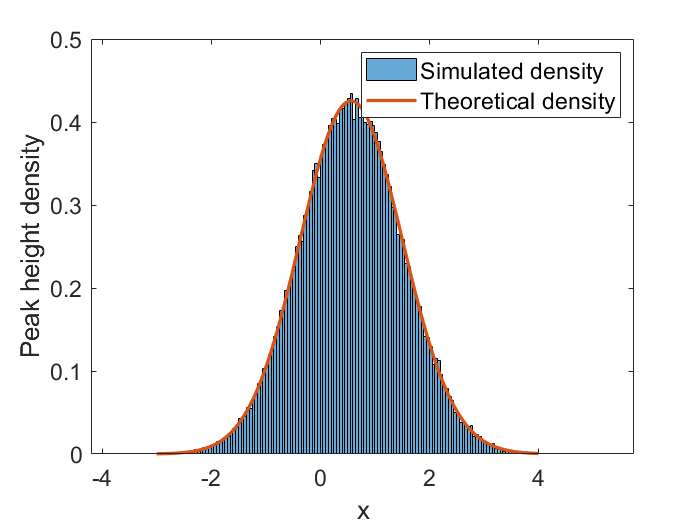}
    \caption{Peak height density}
    \label{fig:1D_linear_PeakHeight}
\end{subfigure}
\caption{1D simulation: $\rho(t)$ and the peak height density of the Gaussian process with non-constant bandwidth.}
\label{fig:1D_bandwdith_peak}
\end{figure}

Since $\nu(t) = 0.5t + 0.1$ is a linear function of $t$, according to Theorem \ref{thm:linear_nu}, both the parameter $\rho(t)$ and the peak height distribution should be independent of $t$. Figure \ref{fig:1D_linear_rho} shows strong evidence that $\rho(t)$ is constant and the simulated value matches the theoretical value. Figure \ref{fig:1D_linear_PeakHeight} demonstrates that the theoretical density curve computed using \eqref{eq:peak_height_1D} at a single location closely resembles the empirical peak height density over the entire domain.

\subsection{3D scale space Gaussian random field}\label{sim:3d_scale}

Theorem \ref{thm:scale_space} states that the peak height distribution of the scale space field does not vary with respect to the scale parameter $\nu$ and location parameter $t$. To validate this, we apply the numerical Kac-Rice Algorithm \ref{alg:Kac-Rice-1} to compare the peak height distribution of a 3D scale space field with different $t$ and $\nu$. The simulation setup considers the Gaussian kernel and three different scenarios, each with a different set of parameters. For each scenario, we perform $10^6$ simulations. The input of the algorithm, which are the variances and covariances of the first and second order derivatives, is derived in \ref{sec:scale_gaussian}. The parameter choices for the three scenarios are listed in Table \ref{tab:sim_para}. 

\begin{table}[h]
\caption{Parameter choices for the three simulation scenarios to validate Theorem \ref{thm:scale_space}. Scenario 1 can be seen as the baseline. The effect of the scale parameter $\nu$ and the location parameter $t$ are reflected in the comparisons between Scenario 1 and 2, and between Scenario 1 and 3, respectively. }
\centering
\begin{tabular}{c
                c
                c
                c
} \toprule
Parameter &  Scenario 1 & Scenario 2 & Scenario 3 \\ \toprule
$t_1$ & 0.5 & 0.5& 1 \\
$t_2$ & 0.5 & 0.5& 0.5 \\
$\nu$  & 0.7 & 0.2& 0.7 \\
\bottomrule
\end{tabular} 
\label{tab:sim_para}
\end{table} 

\begin{figure}
\centering
\includegraphics[width = 200pt,valign=t]{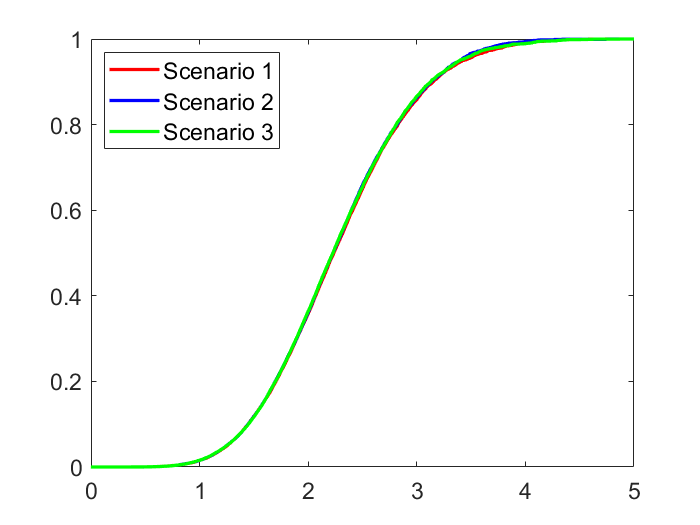}
\label{fig:3D_scale_thm}
\caption{Simulated CDF of the 3D scale space field under the three scenarios described in Table \ref{tab:sim_para}.}
\label{fig:sim_thm_scale}
\end{figure}

As we can see from Figure \ref{fig:sim_thm_scale}, all three scenarios lead to the same peak height distribution. Given the large number of simulations, this provides strong evidence to support Theorem \ref{thm:scale_space}.

\section{Discussion}

\subsection{Peak height distribution in 1D}
In this paper, we have derived the explicit formula for the peak height distribution of non-stationary Gaussian processes in 1D, and demonstrated that the formula has a similar form as that for stationary Gaussian processes while using a different parameterization, and depends on two parameters, $\rho(t)$ and $\tilde{\sigma}(t)$. \citet{Bernoulli} defined the parameter $\kappa$ for isotropic Gaussian random fields because it has the nice property of being invariant to the scaling
of the parameter space and equal to 1 when the covariance is Gaussian. However, here we discovered $\rho(t)$ is a better parameterization, as it has a clear statistical meaning (the conditional correlation of the process and its second derivative given that the gradient is 0) and allows characterizing a larger range of possible processes.  Since $\rho(t)$ is defined as a conditional correlation, we have $|\rho(t)| \leq 1$. When $\rho(t) = \pm1$, the peak height distribution reduces to the Rayleigh distribution.. 

Like $\kappa$, the parameter $\rho(t)$ is also invariant to scaling of the parameter space, but it is equal to $-1/\sqrt{3}$ instead when the covariance is Gaussian. Unlike $\kappa$, $\rho(t)$ is not always positive or negative. Although $\rho(t)$ is non-positive when the process has constant variance, we have provided an example demonstrating it can be positive in general. 

The other parameter $\tilde{\sigma}(t)$ is the scale parameter of the distribution and becomes constant when the process has constant variance. $\tilde{\sigma}(t)$ also has a clear statistical meaning (the conditional variance of the process and
its second derivative given that the gradient is 0) and the mean and standard deviation are proportional to $\tilde{\sigma}(t)$.  

We provide two methods for computing $\rho(t)$ and $\tilde{\sigma}(t)$: either through the moments of the process and its derivatives or through the partial derivatives of the covariance function, depending on which is more convenient in application.  

\subsection{Peak height distribution of the scale space field}

With the help of the peak height density function in 1D, we discover an interesting property of the Gaussian process with non-constant bandwidth and generalize it to the multidimensional scale space field. This property states that the peak height distribution of the scale space field does not vary over the domain. This result is rather surprising, given that the scale space field is nonstationary.

Studying the peak height distribution of the scale space field helps to choose the proper threshold to control the type I error when the signal peaks we aim to detect have different shapes (\citealp{shafie2003}). However, the scale space field is not optimal for detecting peaks that are not rotationally symmetric. In future work, we plan to investigate the Gaussian random field defined on the scale-rotation space and explore the properties of the peak height distribution.

\subsection{Numerical Kac-Rice algorithm}

Apart from the theoretical discoveries, we also provide numerical algorithms for computing the peak height distribution of any smooth Gaussian random field. One inconvenience of the numerical algorithms is that the user needs to compute the variances and covariances of the field and its first and second order derivatives. Since deriving the analytical form for the covariance matrices, like what we have done for the scale space field, involves only taking derivatives of the covariance function \eqref{eq:cov_scale}, the slight inconvenience of extra computation is usually not a big issue in practice. Also, the derivation can be automated by employing symbolic computation tools (e.g. Symbolic Math Toolbox in Matlab and WolframAlpha) as such tools are effective in performing and simplifying complicated differentiation. Alternatively, it is feasible to estimate these variances and covariances empirically from the real or simulated data, as it typically requires a smaller number of data instances to obtain accurate estimates for them compared to estimating the peak height distribution considering not all data instances have local maxima.

Although in practice the numerical methods are convenient and accurate for a sufficiently large number of iterations, we still recommend exploring the properties of the peak height distribution first to better understand its spatial structure and increase the algorithm efficiency. Theorem \ref{thm:scale_space} is a good example. The peak height distribution for the scale space field is independent of the location and scale parameter. This allows us to compute the peak height distribution for a single point and generalize it to the entire domain. 

The amount of time required to generate the field and search for local maxima has been a common issue for evaluating the theoretical properties of the peak height distribution by simulation, especially when the field is multidimensional. We anticipate that the numerical Kac-Rice algorithm will facilitate and accelerate the exploration of the peak height distribution of multidimensional Gaussian random fields.

\section{Acknowledgements}
Y.Z., D.C., S.D. and A.S. were partially supported by NIH grant R01EB026859. Y.Z., D.C. and A.S. were partially supported by NSF grant 1811659. 
\begin{appendices}

\section{}

\begin{proof}[Proof of Lemma \ref{lem:X_cov}]

Following \eqref{def:cov_derv}, the covariance function of $X$ is 
\begin{equation*}
\E[X(t)X(s)] =  h\left(\frac{t+s}{2},(t-s)^2\right). 
\end{equation*}
Taking partial derivatives on both sides, we have
\begin{align*}
    \frac{\partial}{\partial t} \E[X(t)X(s)] = & \frac{1}{2}h'_1\left(\frac{t+s}{2},(t-s)^2\right) + 2(t-s)h'_2\left(\frac{t+s}{2},(t-s)^2\right),\\ 
    \frac{\partial^2}{\partial t^2} \E[X(t)X(s)] = & \frac{1}{4}h''_{11}\left(\frac{t+s}{2},(t-s)^2\right) + 2(t-s)h''_{12}\left(\frac{t+s}{2},(t-s)^2\right) \\
    & + 4(t-s)^2h''_{22}\left(\frac{t+s}{2},(t-s)^2\right) + 2h'_2\left(\frac{t+s}{2},(t-s)^2\right), \\ 
    \frac{\partial^2}{\partial t \partial s} \E[X(t)X(s)] = & \frac{1}{4}h''_{11}\left(\frac{t+s}{2},(t-s)^2\right) - 4(t-s)^2h''_{22}\left(\frac{t+s}{2},(t-s)^2\right) \\
    & - 2 h'_{2}\left(\frac{t+s}{2},(t-s)^2\right), \\
    \frac{\partial^3}{\partial t^2 \partial s} \E[X(t)X(s)] = & \frac{1}{8}h'''_{111}\left(\frac{t+s}{2},(t-s)^2\right) + \frac{1}{2}(t-s)h'''_{112}\left(\frac{t+s}{2},(t-s)^2\right) \\
    & - 2(t-s)^2h'''_{122}\left(\frac{t+s}{2},(t-s)^2\right) - 8(t-s)^3h'''_{222}\left(\frac{t+s}{2},(t-s)^2\right) \\
    & - h''_{12}\left(\frac{t+s}{2},(t-s)^2\right)-12(t-s)h''_{22}\left(\frac{t+s}{2},(t-s)^2\right), \\
    \frac{\partial^4}{\partial t^2 \partial s^2} \E[X(t)X(s)] = & \frac{1}{16}h''''_{1111}\left(\frac{t+s}{2},(t-s)^2\right) - 2(t-s)^2h''''_{1122}\left(\frac{t+s}{2},(t-s)^2\right) \\
    & + 16(t-s)^4h''''_{2222}\left(\frac{t+s}{2},(t-s)^2\right) - h'''_{112}\left(\frac{t+s}{2},(t-s)^2\right) \\
    & + 48(t-s)^2h'''_{222}\left(\frac{t+s}{2},(t-s)^2\right) + 12h''_{22}\left(\frac{t+s}{2},(t-s)^2\right).
\end{align*}
By (5.5.5) in \citet{Adler2007},
\begin{align*}
    \Var(X(t)) & = \left. \E[X(t)X(s)] \right|_{s = t}= h_0(t), \\
    \Var(X'(t)) &=  \left. \frac{\partial^2}{\partial t \partial s} \E[X(t)X(s)] \right|_{s = t} = \frac{1}{4}h''_{11}(t) - 2h'_2(t),\\
    \Var(X''(t)) & = \left. \frac{\partial^4}{\partial t^2 \partial s^2} \E[X(t)X(s)] \right|_{s = t} = \frac{1}{16}h''''_{1111}(t) - h'''_{112}(t) + 12h''_{22}(t), \\
    \E[X(t)X'(t)] &=\left. \frac{\partial}{\partial t} \E[X(t)X(s)] \right|_{s = t}= \frac{1}{2}h'_1(t),\\
    \E[X(t)X''(t)] &=\left. \frac{\partial^2}{\partial t^2} \E[X(t)X(s)] \right|_{s = t}= \frac{1}{4}h''_{11}(t) + 2h'_2(t), \\
    \E[X'(t)X''(t)] &=\left. \frac{\partial^3}{\partial t^2 \partial s} \E[X(t)X(s)] \right|_{s = t}= \frac{1}{8}h'''_{111}(t) - h''_{12}(t).
\end{align*}

\end{proof}

\begin{proof}[Proof of Theorem \ref{thm:peak_height_1D}]

Let $Z(t) = X(t)/\sigma(t)$ so that $Z(t)$ has unit variance. To apply the Kac-Rice formula \eqref{eq:point} to compute the peak height distribution, we need the variances and covariances of the first and second order derivatives. Suppose the random vector $(Z(t), Z'(t), Z''(t))$ has the covariance structure
\begin{equation*}
    \begin{split}
        \begin{pmatrix}
            {\rm Var}(Z(t)) & \E[Z(t)Z'(t)] & \E[Z(t)Z''(t)] \\
            \E[Z(t)Z'(t)] & {\rm Var}(Z'(t)) & \E[Z'(t)Z''(t)] \\
            \E[Z(t)Z''(t)] & \E[Z'(t)Z''(t)] & {\rm Var}(Z''(t))
        \end{pmatrix}= \begin{pmatrix}
            1 & 0 & -\lambda_1(t) \\
            0 & \lambda_1(t) & r_1(t) \\
            -\lambda_1(t) & r_1(t) & \lambda_2(t)
        \end{pmatrix},
    \end{split}
\end{equation*}  
where $Z'(t)$ and $Z''(t)$ are the first and second order derivatives of $Z(t)$. $\E[Z(t)Z'(t)] = 0$ and ${\rm Var}(Z'(t)) = -\E[Z(t)Z''(t)]$ can be derived from taking the first and second derivatives on both sides of $\Var(Z(t)) = 1$.

With a minor abuse of notation, we use $\rho$, $\lambda_1$, $\lambda_2$, $r_1$ to represent $\rho(t)$, $\lambda_1(t)$, $\lambda_2(t)$ and $r_1(t)$, and $\sigma_t$, $\sigma_t'$, $\sigma_t''$, $\tilde{\sigma}_t$ to represent $\sigma(t)$, $\sigma'(t)$, $\sigma''(t)$, $\tilde{\sigma}(t)$. Following this notation, we have $X'(t)=\sigma_t'Z(t) + \sigma_tZ'(t)$ and $X''(t)=\sigma_t''Z(t) + 2\sigma_t'Z'(t) + \sigma_tZ''(t)$. $X(t)$, $X'(t)$ and $X''(t)$ are Gaussian distributed with mean 0, and variance-covariance matrix
\begin{equation*}
    \begin{split}
        &\quad \begin{pmatrix}
            {\rm Var}(X(t)) & \E[X(t)X'(t)] & \E[X(t)X''(t)] \\
            \E[X(t)X'(t)] & {\rm Var}(X'(t)) & \E[X'(t)X''(t)] \\
            \E[X(t)X''(t)] & \E[X'(t)X''(t)] & {\rm Var}(X''(t))
        \end{pmatrix}\\
    & = \begin{pmatrix}
            \sigma_t^2 & \sigma_t\sigma_t' & \sigma_t\sigma_t''- \sigma_t^2\lambda_1\\
            \sigma_t\sigma_t' & \sigma_t'^2+ \sigma_t^2\lambda_1 & \sigma_t'\sigma_t''+ \sigma_t\sigma_t'\lambda_1 + \sigma_t^2r_1\\
            \sigma_t\sigma_t''- \sigma_t^2\lambda_1 & \sigma_t'\sigma_t''+ \sigma_t\sigma_t'\lambda_1 + \sigma_t^2r_1 & \sigma_t''^2 + (4\sigma_t'^2-2\sigma_t\sigma_t'')\lambda_1 + \sigma_t^2\lambda_2 + 4\sigma_t \sigma'_tr_1
        \end{pmatrix}\\
    & = \begin{pmatrix}
    a & d & f \\ 
    d & b & e \\ 
    f & e & c
    \end{pmatrix},
    \end{split} 
\end{equation*} 
where $a$, $b$, $c$, $d$, $e$, $f$ are some constants for fixed $t$.  Therefore
\begin{equation*}
\E{[X(t),X''(t)|X'(t) = 0]} = (0,0),
\end{equation*}
and
\begin{align*}
\Var{(X(t),X''(t)|X'(t)=0)} & = \begin{pmatrix}
a & f\\ 
f & c
\end{pmatrix} - \begin{pmatrix}
d \\ 
e
\end{pmatrix}\frac{1}{b}\begin{pmatrix}
d & e
\end{pmatrix}= \begin{pmatrix}
a - \frac{d^2}{b} & f-\frac{de}{b} \\ 
f - \frac{de}{b} & c - \frac{e^2}{b}
\end{pmatrix}.
\end{align*}
According the definition, we have
\begin{align*}
 \rho &:= \Cor(X(t),X''(t)|X'(t) = 0) = \frac{fb-de}{\sqrt{(ab-d^2)(bc-e^2)}} \\
 & = - \frac{\sigma_t^2\lambda_1^2 + (2\sigma_t'^2 - \sigma_t\sigma_t'')\lambda_1 + \sigma_t\sigma_t'r_1}{\sqrt{ \splitfrac{\lambda_1(\sigma_t^2(3\sigma_t'^2 - 2\sigma_t\sigma_t'')\lambda_1^2 + (\sigma_t\sigma_t'' - 2\sigma_t')^2\lambda_1+ \sigma_t^4\lambda_1\lambda_2 + 2\sigma_t^3\sigma_t'\lambda_1 r_1}{ + \sigma_t^2\sigma_t'^2\lambda_2 - \sigma_t^4r_1^2 + (4\sigma_t\sigma_t'^3 - 2\sigma_t^2\sigma_t'\sigma_t'')r_1)} }},
 \end{align*}
and
\begin{align*}
 \tilde{\sigma}_t^2 &:= \Var(X(t)|X'(t) = 0) = a - \frac{d^2}{b} = \frac{\sigma_t^4\lambda_1}{\sigma_t'^2 + \sigma_t^2\lambda_1}.
\end{align*}

The peak height density is given by the derivative of \eqref{eq:point}
\[
f_t(x) = \frac{\E[|X''(t)|\mathbbm{1}_{\{X''(t) < 0\}} | X(t)=x, X'(t)=0]p_{X(t)}(x|X'(t)=0)}{\E[|X''(t)|\mathbbm{1}_{\{X''(t) < 0\}} | X'(t)=0]}.
\]
First, we compute the denominator
\begin{align*}
& \E{[|X''(t)|\mathbbm{1}_{\{X''(t) < 0\}}|X'(t) = 0]} = -\int_{-\infty}^{0} \frac{x}{\sqrt{c-\frac{e^2}{b}}} \phi\left(\frac{x}{\sqrt{c - \frac{e^2}{b}}}\right) dx = \frac{\sqrt{c - \frac{e^2}{b}}}{\sqrt{2\pi}}.
\end{align*}

Similarly, $X''(t)|X(t) = x, X'(t) = 0$ is Gaussian distributed with mean
\begin{align*}
 & \E[X''(t)|X(t) = x, X'(t) = 0]=\begin{pmatrix}
 f & e
 \end{pmatrix} \begin{pmatrix}
 a & d\\ 
 d & b
 \end{pmatrix}^{-1} \begin{pmatrix}
 x \\ 
 0  
 \end{pmatrix} =  \frac{fb - ed}{ab - d^2} x=  -\eta x,
\end{align*}
and variance
\begin{align*}
 \Var(X''(t)|X(t) = x, X'(t) = 0)=& c - \begin{pmatrix}
f & e
\end{pmatrix} \begin{pmatrix}
a & d \\ 
d & a 
\end{pmatrix}^{-1}\begin{pmatrix}
f \\
e
\end{pmatrix} \\
= & c - \frac{f^2b - 2fed + ae^2}{ab-d^2} =  c - \xi,
\end{align*}
where $\eta:= (ed - fb)/(ab- d^2)$ and $\xi:= (f^2b - 2fed + ae^2)/(ab-d^2)$. Now the numerator of the Kac-Rice formula becomes
\begin{align*}
& \E[|X''(t)|\mathbbm{1}_{\{X''(t) < 0\}}|X(t) = x, X'(t) = 0]\frac{1}{ \tilde{\sigma}_t}\phi\left(\frac{x}{ \tilde{\sigma}_t}\right)  \\ 
= & -\frac{1}{ \tilde{\sigma}_t}\phi\left(\frac{x}{ \tilde{\sigma}_t}\right)\int_{-\infty}^{0} \frac{y}{\sqrt{c-\xi}}\phi\left(\frac{y+\eta x}{\sqrt{c-\xi}}\right) dy  \\ 
= &  \eta x\Phi\left(\frac{\eta x}{\sqrt{c-\xi}}\right)\frac{1}{\tilde{\sigma}_t}\phi\left(\frac{x}{\tilde{\sigma}_t}\right) + \sqrt{c-\xi}\phi\left(\frac{\eta x}{\sqrt{c-\xi}}\right)\frac{1}{\tilde{\sigma}_t}\phi\left(\frac{x}{\tilde{\sigma}_t}\right).
\end{align*}
Combine the results above and consider \eqref{def: psi}, we have
\begin{align*}
f_t(x) &= \frac{\eta x\Phi\left(\frac{\eta x}{\sqrt{c-\xi}}\right)\frac{1}{\tilde{\sigma}_t}\phi\left(\frac{x}{\tilde{\sigma}_t}\right) + \sqrt{c-\xi}\phi\left(\frac{\eta x}{\sqrt{c-\xi}}\right)\frac{1}{\tilde{\sigma}_t}\phi\left(\frac{x}{\tilde{\sigma}_t}\right)}{\frac{\sqrt{c - \frac{e^2}{b}}}{\sqrt{2\pi}}} \\ 
& =  \frac{1}{\tilde{\sigma}_t}\phi\left(\frac{x}{\tilde{\sigma}_t}\right) \sqrt{\frac{2\pi(c-\xi)}{c-\frac{e^2}{b}}} \left[\phi\left(\frac{\eta x}{\sqrt{c-\xi}}\right) +  \frac{\eta x}{\sqrt{c-\xi}} \Phi\left(\frac{\eta x}{\sqrt{c-\xi}}\right)\right] \\ 
& = \frac{1}{\tilde{\sigma}_t}\phi\left(\frac{x}{\tilde{\sigma}_t}\right) \sqrt{\frac{2\pi(c-\xi)}{c-\frac{e^2}{b}}} \psi\left(\frac{\eta x}{\sqrt{c-\xi}}\right).
\end{align*}
Using the fact that
\begin{align*}
    \xi = \tilde{\sigma}_t^2 \eta^2 + \frac{e^2}{b} = \sigma_t^2\lambda_1^2 - 2(\sigma_t\sigma_t'' - 2\sigma_t'^2)\lambda_1 + \frac{r_1^2\sigma_t^2}{\lambda_1} + 4\sigma_t\sigma_t' + \sigma_t''^2,
\end{align*}
we have
\begin{align*}
    1 - \frac{(c-\xi)}{c-\frac{e^2}{b}} & = \frac{\xi - \frac{e^2}{b}}{c - \frac{e^2}{b}} = \frac{\tilde{\sigma}_t^2 \eta^2}{c - \frac{e^2}{b}} = \frac{(fb - dc)^2}{(ab-d^2)(bc-e^2)} = \rho^2,
\end{align*}
and 
\begin{align*}
    1 + \tilde{\sigma}_t^2\frac{\eta^2}{c - \xi} & = \frac{c-\frac{e^2}{b}}{c - \xi} = \frac{1}{1-\rho^2}.
\end{align*}
Therefore, the peak height density can be further simplified as  
\begin{align*}
    f_t(x) = \frac{1}{\tilde{\sigma}_t}\phi\left(\frac{x}{\tilde{\sigma}_t}\right) \sqrt{2\pi(1-\rho^2)}\psi\left(\frac{-\rho x}{\sqrt{1-\rho^2}\tilde{\sigma}_t}\right).
\end{align*}

The mean and variance of the peak height can be derived from its density function. Using the properties of Gaussian integrals (\citealp{Owen1980}), we have
\begin{align*}
    \int_{-\infty}^{\infty} xf_t(x) dx
    = & \int_{-\infty}^{\infty} \sqrt{2\pi(1-\rho^2(t))} \frac{x}{\tilde{\sigma}(t)}\phi\left(\frac{x}{\tilde{\sigma}(t)}\right) \phi\left(\frac{-\rho(t)x}{\sqrt{1-\rho^2(t)}\tilde{\sigma}(t)}\right)dx \\
    & - \int_{-\infty}^{\infty} \sqrt{2\pi(1-\rho^2(t))} \frac{x}{\tilde{\sigma}(t)}\phi\left(\frac{x}{\tilde{\sigma}(t)}\right) \frac{\rho(t)x}{\sqrt{1-\rho^2(t)}\tilde{\sigma}(t)}\Phi\left(\frac{-\rho(t)x}{\sqrt{1-\rho^2(t)}\tilde{\sigma}(t)}\right)dx \\
    = & - \sqrt{2\pi}\rho(t)\tilde{\sigma}(t) \int_{-\infty}^{\infty} x^2\phi(x)\Phi\left(\frac{-\rho(t)x}{\sqrt{1-\rho^2(t)}}\right) dx \\
    = & -\sqrt{\frac{\pi}{2}}\rho(t)\tilde{\sigma}(t),
\end{align*}
and
\begin{align*}
    \int_{-\infty}^{\infty} x^2f_t(x) dx
    = & \sqrt{2\pi(1-\rho^2(t))}\tilde{\sigma}^2(t) \int_{-\infty}^{\infty} x^2\phi(x)\phi\left(\frac{-\rho(t)x}{\sqrt{1-\rho^2(t)}}\right) dx \\
    & - \sqrt{2\pi}\rho(t)\tilde{\sigma}^2(t) \int_{-\infty}^{\infty} x^3\phi(x)\Phi\left(\frac{-\rho(t)x}{\sqrt{1-\rho^2(t)}}\right) dx \\
    = & (1-\rho^2(t))^2\tilde{\sigma}^2(t) - \rho^2(t)(\rho^2(t) -3)\tilde{\sigma}^2(t) \\
    = & (\rho^2(t) + 1)\tilde{\sigma}^2(t).
\end{align*}
The mean and variance are direct results of the first and second moments.

\end{proof}

\begin{proof}[Proof of Proposition \ref{prop:boundary}]

Since the derivation for $\rho(t) = -1$ and $\rho(t) = 1$ follows a similar approach, here we only illustrate the method for $\rho(t) = -1$. With $\rho(t) = \Cor(X(t), X''(t)|X'(t) = 0) = -1$, the numerator of the Kac-Rice formula \eqref{eq:point} becomes
\begin{align*}
    & \E[|X''(t)|\mathbbm{1}_{\{X''(t) < 0 \}}\mathbbm{1}_{\{X(t)>u\}} |  X'(t)=0] \\
    = & \E[|X''(t)|\mathbbm{1}_{\{X''(t)< - \sqrt{\Var(X''(t)|X'(t)=0)/\Var(X(t)|X'(t)=0)}u\}} |  X'(t)=0] \\
    = &\int_{-\infty}^{-\sqrt{\Var(X''(t)|X'(t)=0)/\tilde{\sigma}(t)}u} \frac{x}{\sqrt{\Var(X''(t)|X'(t)=0)}}\phi\left(\frac{x}{\sqrt{\Var(X''(t)|X'(t)=0)}}\right) dx\\
    = & \sqrt{\Var(X''(t)|X'(t)=0)}\phi\left(\frac{u}{\tilde{\sigma}(t)}\right).
\end{align*}   
The denominator is the same as the non-boundary case
\begin{align*}
& \E{[|X''(t)|\mathbbm{1}_{\{X''(t) < 0\}}|X'(t) = 0]} = \frac{\sqrt{\Var(X''(t)|X'(t)=0)}}{\sqrt{2\pi}}.
\end{align*}
The peak height distribution  
\begin{align*}
    F_t(u) &= \frac{\E[|X''(t)|\mathbbm{1}_{\{X''(t) < 0 \}}\mathbbm{1}_{\{X(t)>u\}} |  X'(t)=0]}{\E{[|X''(t)|\mathbbm{1}_{\{X''(t) < 0\}}|X'(t) = 0]}} = 2\pi\phi\left(\frac{u}{\tilde{\sigma}(t)}\right).
\end{align*}
Taking derivative, we obtain
\begin{equation*}
   f_t(x) = \frac{\sqrt{2\pi}x}{\tilde{\sigma^2}(t)}\phi\left(\frac{x}{\tilde{\sigma}(t)}\right).  
\end{equation*}
\end{proof}

\begin{proof}[Proof of Theorem \ref{thm:linear_nu}]
Take the derivative on both sides of \eqref{def:non_constant_band}, we obtain
\begin{equation*}
X'(t) = \int_{-\infty}^{\infty} \left[-\frac{\nu'(t)}{2\nu^{3/2}(t)}k\left(\frac{t-s}{\nu(t)}\right) + k'\left(\frac{t-s}{\nu(t)}\right)\frac{\nu(t)-(t-s)\nu'(t)}{\nu^{5/2}(t)} \right] d B(s).
\end{equation*}
Using the property of Wiener Integral and applying the change of variable $u = (t-s)/\nu(t)$, we have
\begin{align*}
\Var{(X'(t))} &= \int_{-\infty}^{\infty}\left[-\frac{\nu'(t)}{2\nu^{3/2}(t)}k\left(\frac{t-s}{\nu(t)}\right) + k'\left(\frac{t-s}{\nu(t)}\right)\frac{\nu(t)-(t-s)\nu'(t)}{\nu^{5/2}(t)} \right]^2 ds \\
& = \int_{-\infty}^{\infty} \left[-\frac{\nu'(t)}{2\nu(t)}k(u) + \frac{1-u\nu'(t)}{\nu(t)}k'(u)\right]^2 du \\
& = \frac{1}{\nu^2(t)}\int_{-\infty}^{\infty} \left[-\frac{1}{2}\nu'(t)k(u) + (1-u\nu'(t))k'(u) \right]^2 du\\
& = \frac{1}{\nu^2(t)}C_1,
\numberthis
\label{prop2: lambda_1}
\end{align*}
where $C_1$ is a constant when $\nu(t)$ is a linear function of t.

Taking the derivative again, we have
\begin{align*}
X''(t) = & \int_{-\infty}^{\infty} -\frac{1}{2}\left[\frac{2\nu''(t)\nu^{3/2}(t)-3\nu'^2(t)\nu^{1/2}(t)}{2\nu^3(t)}k\left( \frac{t-s}{\nu(t)}\right) \right. \\
& + \frac{\nu'(t)\nu(t)-(t-s)\nu'^2(t)}{\nu^{7/2}(t)} k'\left(\frac{t-s}{\nu(t)} \right) +  \frac{(\nu(t)-(t-s)\nu'^2(t))}{\nu^{9/2}(t)} k''\left(\frac{t-s}{\nu(t)}\right) \\
& \left. - \frac{2(t-s)\nu''(t)\nu^{5/2}(t)+5\nu'(t)\nu^{5/2}(t)-5(t-s)\nu'^2(t)\nu^{3/2}(t)}{\nu^5(t)} \right] dB(s).
\end{align*}
Again, let $u = (t-s)/\nu(t)$ and use the fact $\nu''(t) = 0$ and $\nu'(t)$ independent of $t$  when $\nu(t)$ is linear,
\begin{align*}
\Var{(X''(t))} & = \frac{1}{\nu^4(t)}C_2
\numberthis
\label{prop2: lambda_2}
\end{align*}
for some constant $C_2$. Similar to \eqref{prop2: lambda_1} and \eqref{prop2: lambda_2},
\begin{align*}
\E{[X'(t)X''(t)]} & = \frac{1}{\nu^3(t)}C_3,
\numberthis
\label{prop2: r_1}
\end{align*}
where $C_3$ is a constant when $\nu(t)$ is linear. Combining \eqref{prop2: lambda_1}, \eqref{prop2: lambda_2} and \eqref{prop2: r_1},
\begin{align*}
\rho(t) &= -\frac{\Var{(X'(t))}}{\sqrt{\Var{(X''(t))} - \frac{\E[X'(t)X''(t)]^2}{\Var{(X'(t))}}}} = \frac{C_1}{\sqrt{C_2-\frac{C_3^2}{C_1}}}.
\end{align*}

\end{proof}

\begin{proof}[Proof of \eqref{eq:rho_gauss}]
We first compute $\Var{(X'(t))}$, $\Var{(X''(t))}$ and $\E[X'(t)X''(t)]$. Taking the first and second order derivative of $X(t)$, we obtain
\begin{equation*}
    \begin{split}
    X'(t) = \sqrt{2}\pi^{1/4}\int_{-\infty}^\infty \left[-\frac{\nu'(t)}{2\nu^{3/2}(t)} - \frac{t-s}{\nu^{5/2}(t)} + \frac{3(t-s)^2\nu'(t)}{\nu^{7/2}(t)}\right]\phi\left(\frac{t-s}{\nu(t)}\right)\,dB(s),
    \end{split}
\end{equation*} 
and 
\begin{align*}
X''(t) = &\sqrt{2}\pi^{1/4}\int_{-\infty}^\infty \left[ \frac{\nu'^2(t)}{\nu^{13/2}(t)}(t-s)^4 - \frac{2\nu'(t)}{\nu^{11/2}(t)}(t-s)^3 + \left(\frac{\nu''(t)}{\nu^{7/2}(t)} + \frac{2-9\nu'^2(t)}{2\nu^{9/2}(t)} \right) (t-s)^2 \right. \\
 &\left. + \frac{5\nu'(t)}{\nu^{7/2}(t)}(t-s) -\frac{\nu''(t)}{2\nu^{3/2}(t)} + \frac{3\nu'^2(t)}{4\nu^{5/2}(t)}-\frac{1}{\nu^{5/2}(t)} \right] \phi\left(\frac{t-s}{\nu(t)} \vphantom{\frac{1}{\nu^{5/2}(t)}} \right)\,dB(s).
\end{align*}
$\Var{(X'(t))}$, $\Var{(X''(t))}$ and $\E[X'(t)X''(t)]$ can be computed using the property of Wiener integral:
\begin{align*}
\Var{(X'(t))} &= \frac{1}{\pi^{1/2}} \int_{-\infty}^{\infty} \left[-\frac{\nu'(t)}{2\nu^{3/2}(t)} - \frac{t-s}{\nu^{5/2}(t)} + \frac{3(t-s)^2\nu'(t)}{\nu^{7/2}(t)}\right]^2 e^{-\frac{(t-s)^2}{\nu^2(t)}} ds
= \frac{1+\nu'^2(t)}{2\nu^2(t)}.
\end{align*}
Similarly
\begin{align*}
\Var{(X''(t))} & = \frac{51\nu'^4(t)+68\nu'^2(t)+12}{16\nu^4(t)} - \frac{\nu''(t)(3\nu'^2(t)-2)}{2\nu^3(t)} + \frac{\nu''^2(t)}{2\nu^2(t)},
\end{align*}
and
\begin{align*}
\E{[X'(t)X''(t)]} = -\frac{\nu'^3(t)+ \nu'(t)}{2\nu^3(t)} + \frac{\nu'(t)\nu''(t)}{2\nu^2(t)}.
\end{align*}
Combine the results above and apply \eqref{eq:rho_constant_sigma_vcov}, we have
\begin{align*}
\rho(t) &= -\frac{1}{2}\frac{\nu'^2(t)+1}{ \sqrt{ \splitfrac{\frac{51}{16}\nu'^4(t)+(\frac{68}{16}-\frac{3\nu''(t)\nu(t)}{2})\nu'^2(t)+\frac{3}{4}+\nu''(t)\nu(t)+\frac{1}{2}\nu''^2(t)\nu^2(t)}{-\frac{(-\nu'^3(t)+(\nu''(t)\nu(t)-1)\nu'^2(t))}{2(\nu'^2(t)+1)}}}} \\
& = -\frac{1}{2}\sqrt{\frac{(\nu'^2(t)+1)^{3}}{\splitfrac{\frac{43}{16}\nu'^6(t) + (\frac{103}{16}-\frac{1}{2}\nu''(t)\nu(t))\nu'^4(t) + \frac{9}{2}\nu'^2(t)+\frac{3}{4}}{+\nu''(t)\nu(t)+\frac{1}{2}\nu''^2(t)\nu^2(t)}}}.
\end{align*}
\end{proof}

\begin{proof}[Proof of Theorem \ref{thm:scale_space}]
For convenience, apply the change of variable $v = -\log\nu$. The covariance function 
\begin{equation*}
\E{[X(t,v)X(\tilde{t},\tilde{v})]} = e^{N(v+\tilde{v})/2} \int k((s-t)e^{v}) k((s-\tilde{t})e^{\tilde{v}}) ds.
\end{equation*}

At a fixed point, the first-order partial derivatives of $X$ are Gaussian distributed with mean 0. The variances and covariances can be derived by applying (5.5.5) in \citet{Adler2007}
\begin{align*}
\Var{\left(\frac{\partial X}{\partial t}\right)} &= e^{2v}\int \dot{k}(s)\dot{k}(s)^T dh,\\
\Var{\left(\frac{\partial X}{\partial v}\right)} &= \int \left(s^T\dot{k}(s) + \frac{N}{2}k(s)\right)^2 ds,\\
\E{\left[\frac{\partial X}{\partial t}\frac{\partial X}{\partial v}\right]} &  = 0,
\end{align*}
where $\dot{k}(s)$ denotes $\partial k(s)/\partial s$. We define a random vector $G = (G_i)_{1 \leq i \leq N+1}$ such that

\begin{equation*}
 G_i = \begin{cases}
 \frac{1}{e^v} \frac{\partial X}{\partial t_i} & i \leq N\\
 \frac{\partial X}{\partial v} & i = N+1
 \end{cases}.
 \end{equation*} 
Note that the distribution of $G$ is independent of $(t,v)$.

The second-order partial derivatives of $X$, i.e. entries of the Hessian matrix $\nabla^2 X$ are also Gaussian distributed with mean 0. 
\begin{align*}
\frac{\partial^2 X}{\partial t \partial t^T} = e^{2v}A, \quad \frac{\partial^2 X}{\partial t \partial v} = e^{v}B, \quad \frac{\partial^2 X}{\partial v^2} = C,
\end{align*}
where the distribution of $A$ is proportional to $\GOI{(1/2)}$ (\citealp{Bernoulli}). The variance covariance matrix for vector $B$ and scalar $C$ are
\begin{align*}
\Var{(B)} &= \int \left(\ddot{k}(s)s + \left(\frac{N}{2} + 1\right)\dot{k}(s) \right) \left(\ddot{k}(s)s + \left(\frac{N}{2} + 1\right)\dot{k}(s)\right)^Tds,\\
\Var{(C)} &  =  \int \left( s^T\ddot{k}(s)s + (N+1)s^T\dot{k}(s) + \frac{N^2}{4}k(s) \right)^2 ds,
\end{align*}
and the covariance between $A$, $B$ and $C$ does not depend on the parameter $(t,v)$.

\begin{equation*}
\nabla^2 X(t,v) = \begin{bmatrix}
  (\frac{\partial^2 X}{\partial t_i \partial t_j})_{1\leq i,j \leq N}
 &  ( \frac{\partial^2 X}{\partial t_i \partial v})_{1\leq i\leq N} \\
  ( \frac{\partial^2 X}{\partial v \partial t_j})_{1\leq j\leq N} & \frac{\partial^2 X}{\partial v^2} 
\end{bmatrix} 
= \begin{bmatrix}
e^{2v}A & e^{v}B \\
e^{v}B & C
\end{bmatrix},
\end{equation*}
and define the matrix $H$ as 
\begin{equation*}
H = \begin{bmatrix}
A & B \\
B & C
\end{bmatrix}.
\end{equation*}
The distribution of $H$ is again independent of $(t,v)$ like $G$. Similarly, the covariance between $X$ and the entries of $G$, the covariance between $X$ and the entries of $H$, and the covariance between the entries of $G$ and the entries of $H$ are all independent of $(t,v)$. Therefore, the joint distribution of $(X,G,H)$ is independent of $(t,v)$. 

Next, we have
\begin{align*}
\det \nabla^2 X(t,v) &= \det{(e^{2v}A)}\det{(C - e^vB^Te^{-2v}A^{-1}e^vB)}\\
& = e^{2Nv}\det{(A)}\det{(C - B^TA^{-1}B)}\\
& = e^{2Nv}\det{(H)}.
\end{align*} 
Note that the determinant of the $k$th principal minor of $\nabla^2 X(t,v)$ is equal to a positive number times the determinant of the $k$th principal minor of $H$. By Sylvester's criterion, $\nabla^2 X(t,v)$ being negative definite is equivalent to $H$ being negative definite. Applying the Kac-Rice formula \eqref{eq:point}, the peak height distribution of $X(t,v)$ for any fixed point
\begin{equation*}
\begin{split}
F_{t,v}(u)&=\frac{\E[|{\rm det} \nabla^2 X(t,v)|\mathbbm{1}_{\{\nabla^2 X(t,v) \prec 0\}}\mathbbm{1}_{\{X(t,v)>u\}} | \nabla X(t,v)=0\}}{\E\{|{\rm det} \nabla^2 X(t,v)|\mathbbm{1}_{\{\nabla^2 X(t,v) \prec 0\}} | \nabla X(t,v)=0]} \\
& = \frac{\int_{x}^{\infty}\E[|{\rm det} \nabla^2 X(t,v)|\mathbbm{1}_{\{\nabla^2 X(t,v) \prec 0\}}| X(t,v) = x, \nabla X(t,v)=0] \phi(x) dx}{\E[|{\rm det} \nabla^2 X(t,v)|\mathbbm{1}_{\{\nabla^2 X(t,v) \prec 0\}} | \nabla X(t,v)=0]} \\
& = \frac{e^{2Nv}\E[|{\rm det} H|\mathbbm{1}_{\{H \prec 0\}}\mathbbm{1}_{\{X>u\}} | G=0]}{e^{2Nv}\E[|{\rm det} H|\mathbbm{1}_{\{H \prec 0\}} | G=0]} \\
& = \frac{\E[|{\rm det} H|\mathbbm{1}_{\{H \prec 0\}}\mathbbm{1}_{\{X>u\}} | G=0]}{\E[|{\rm det} H|\mathbbm{1}_{\{H \prec 0\}} | G=0]}.
\end{split}
\end{equation*} 
\end{proof}

\end{appendices}

\bibliographystyle{myjmva}
\bibliography{ref_grf}

\end{document}